%% file: main.tex
\documentclass[a4paper,fleqn]{cas-sc}

\usepackage[authoryear,longnamesfirst]{natbib}
\usepackage{float}
\usepackage{tikz}
\usepackage[toc,page]{appendix}
\usetikzlibrary{arrows.meta, calc}
\usepackage{amsfonts}
\usepackage{amsmath}
\usepackage{amssymb}
\usepackage{amsthm}
\usepackage{placeins}

\theoremstyle{plain}
\newtheorem{theorem}{Theorem}

\theoremstyle{definition}

\def\tsc#1{\csdef{#1}{\textsc{\lowercase{#1}}\xspace}}
\tsc{WGM}
\tsc{QE}
\tsc{EP}
\tsc{PMS}
\tsc{BEC}
\tsc{DE}
\begin{document}

\let\WriteBookmarks\relax
\def\floatpagepagefraction{1}
\def\textpagefraction{.001}
\shorttitle{Observational Indistinguishability and the Beginning of the Universe}
%\shortauthors{[Redacted for Blind Review]}
%\author[1]{[Redacted for Blind Review]}[orcid=0000-0001-8008-4969]
\shortauthors{Daniel Linford}
\author[1]{Daniel Linford}[orcid=0000-0001-8008-4969]
%\begin{frontmatter}

% Address/affiliation
\affiliation[1]{organization={Department of Philosophy \& Religious Studies},
    addressline={Old Dominion University}, 
    city={Norfolk},
    postcode={23906}, 
    state={Virginia},
    country={United States}}

\title [mode = title]{Observational Indistinguishability and the Beginning of the Universe} 

\maketitle

\begin{abstract}
    Can we infer whether all of physical reality began to exist? Several novel results are offered suggesting a negative verdict. First, two necessary (but not necessarily sufficient) conditions are offered for a cosmic beginning. Second, I discuss three theorems which show that, in a wide range of spacetimes, observers cannot infer the global structure of their spacetime; one of the three shows that spacetimes satisfying a condition necessary for a cosmic beginning are observationally indistinguishable from spacetimes which do not satisfy that condition. Next, I offer two extensions to the aforementioned theorems. The two extensions show that, in an equally wide range of spacetimes, observers cannot collect sufficient data to determine whether the application conditions for the class singularity theorems are satisfied or whether their spacetime satisfies the two necessary conditions for a cosmic beginning. Lastly, a reply is offered to the objection that the skeptical consequences of the aforementioned theorems can be overcome with induction. Importantly, all past singular dust FLRW spacetimes are observationally indistinguishable from counterparts which, while sharing a number of important local properties, either do not include a singularity to the past of every point or else do not have the sort of time ordering intuitively required for a cosmic beginning. 
\end{abstract}

\section{Introduction}

Before the 1960s, most physicists assumed that the singularities appearing in standard cosmological models were mere artifacts of idealizing the universe as perfectly homogeneous and isotropic. According to an oft-told narrative, when Stephen Hawking, Roger Penrose, Robert Geroch, and others subsequently demonstrated that singularities are generic features of relativistic spacetimes, they  established that physical reality most likely began with the Big Bang.\footnote{For an overview of the history up through the Hawking-Penrose theorem, see \citep{Earman:1999}. For a critical discussion, see \citep{Senovilla:2021}. For important work in this development, see \citep{Hawking:1965, Hawking:1966, Geroch:1966, Hawking_Penrose:1970, HawkingEllis:1973}.}  Contrary to the popular narrative, the classic singularity theorems cannot establish a cosmic beginning: they show at most that spacetimes satisfying certain physical conditions must contain incomplete geodesics.\footnote{The popular narrative is mistaken in another important respect as well. While the singularity theorems appear to show that merely relaxing isotropy or homogeneity does not regularize spacetime, one may be able to regularize spacetime by introducing a sufficiently strong inhomogeneity \citep{Senovilla:1990, Senovilla:1993, Senovilla:1996, Senovilla:2007, Senovilla:2021}.} Since a spacetime with incomplete geodesics can also contain complete timelike geodesics, such theorems fall short of establishing that all of physical reality began in the Big Bang. In order to establish a common beginning for all of spacetime, we need to know much more about spacetime's global structure than singularity theorems, themselves, are capable of establishing. Meanwhile, an important result -- first conjectured by David Malament (\citeyear{Malament:1977}) and later proved by J.B. Manchak (\citeyear{Manchak:2009}) -- suggests we will never know the global structure of the spacetime we inhabit. Taking these results together, one may wonder whether a beginning of all of physical reality can be established.

Here, several new results are presented concerning the observational indistinguishability of spacetimes satisfying two intuitive conditions for a cosmic beginning. Section \ref{two-conditions-section} provides two necessary conditions for a spacetime to have a beginning in an intuitive sense. Section \ref{obs-ind-spacetimes} clarifies the notion of an observationally indistinguishable spacetime. Section \ref{mmt-section} summarizes the Malament-Manchak theorems and sketches their proofs and discusses a theorem, previously established by Manchak, entailing that any spacetime with a beginning is observationally indistinguishable from a spacetime without a beginning.  Section \ref{corollaries-section} then presents two additional corollaries to the Malament–Manchak Theorem. According to the corollaries, spacetimes satisfying various global conditions required by the classic singularity theorems have an observationally indistinguishable counterpart that fails to satisfy it and spacetimes with a singularity to the past of every point are observationally indistinguishable from one without a singularity to the past of every point. In section \ref{nemesis-section}, I show that all expanding FLRW spacetimes, with dust stress-energy content and that have a singularity to the past of every point, possess observationally indistinguishable counterparts that, while sharing a number of physically interesting properties, do not include a singularity to the past of every point. The constructions in section \ref{nemesis-section} involve mathematical techniques frequently employed in astrophysical applications, but which have not previously been discussed in connection with observationally indistinguishable spacetimes. Finally, section \ref{inductive-results-section} responds to the objection that the theorems from previous sections may be overcome using induction.  Together, these results suggest a negative verdict for whether there was a beginning for all of spacetime.

\section{Two Conditions for a Cosmic Beginning\label{two-conditions-section}}

In this paper, I am interested in drawing out the implications various theorems, discussed below, have for the beginning of spacetime. To that end, I turn to introducing two necessary conditions that a spacetime with a beginning in any intuitive sense should satisfy. Whatever else the idea of a cosmic beginning might involve, a cosmic beginning involves the universe being finite in an objectively earlier-than direction. Consequently, the beginning of spacetime involves two notions. First, no matter where one is located in spacetime, past history is finite. Second, a direction of time defined everywhere and uniformly throughout spacetime. In the next two subsections, I precisify these intuitive notions.

\subsection{B-Incompleteness}

There is a close relationship between a spacetime being finite to the past and the notion of a spacetime singularity. Singularity theorems typically show only that spacetimes with specific properties are geodesically incomplete.\footnote{A geodesic is complete if and only if its affine parameter extends to $\pm \infty$. A geodesic is incomplete just in case it is not complete. In turn, $(M,g)$ is geodesically complete just in case all of its geodesics are complete, and is otherwise geodesically incomplete.} However, a spacetime that is geodesically complete may include non-geodesic curves that are incomplete. For this reason, a definition of `spacetime singularity' solely in terms of geodesic incompleteness is inadequate. The most sophisticated accounts of spacetime singularities involve a stronger notion of incompleteness, namely, b-incompleteness \citetext{\citealp[36]{Earman:1995}; \citealp{sep-spacetime-singularities}}. To unpack this notion, let's begin by defining a \emph{spacetime} as a pair $(M, g)$, where $M$ is a pseudo-Riemannian manifold with Lorentz signature and $g$ is a metric defined on $M$. Next, consider a causal half-curve $\gamma$ parametrized by $v$ in $(M, g)$, i.e., $\gamma: [0, v_+]\mapsto M$, such that $v_+ \leq +\infty$. The generalized affine length, or g.a.l., reparametrizes any such curve by mappping the curve into a space with Euclidean signature and determines the length of the image of the curve in that new space.\footnote{At each of the tangent spaces at each point in $M$, we can choose a set of four orthonormal basis vectors. In turn, the choice of a set of four orthonormal basis vectors at each point in $M$ defines the \emph{frame field}. Denote the frame field along $\gamma(v)$ as $e^a_i(v)$. Given $e^a_i(0)$, we can define $e^a_0(v)$ for each of the other points along $\gamma(v)$ using parallel transport. The components of any vector $V$ tangent to $\gamma(v)$ can be written as $V^a = \Sigma^4_{i=1}X^i(v)e^a_i(v)$. Using this vector, one can define the generalized affine length g.a.l. along $\gamma(v)$, i.e., $\text{g.a.l.} = \int_0^{v _+}\sqrt{\sum_{i=1}^4 \left(X^i(v^*)\right)^2}dv^*$. The g.a.l. is not, itself, physically significant because it is not invariant with respect to our choice of $e^a_i(v)$. However, whether the g.a.l. is finite or infinite is invariant; hence, whether the g.a.l. is finite or infinite is physically meaningful.}  While the precise value of the g.a.l. is not physically meaningful, whether the g.a.l. is finite or infinite is physically meaningful. A causal curve $\gamma$ is \emph{inextendible} when $\gamma$ cannot be made longer in the past or future direction. When the g.a.l. is finite for an inextendible timelike or null curve, we can think of spacetime as having only finite extension in some direction. A half-curve $\gamma$ is \emph{b-complete} just in case $\gamma$ extends to infinite g.a.l. In turn, $\gamma$ is \emph{b-incomplete} just in case $\gamma$ is not b-complete. A spacetime $(M, g)$ is b-complete just in case all of the causal curves in $(M, g)$ are b-complete. A spacetime that is not b-complete is b-incomplete.

A b-incomplete spacetime can fail to have a beginning in any intuitively compelling sense. For example, Minkowski spacetime with one point removed is b-incomplete, yet includes geodesics that extend to past time-like or null infinity. To have a beginning, we should exclude the possibility that any geodesic extends to past timelike or null infinity. A spacetime $(M, g)$ is \emph{everywhere past b-incomplete} if, for every point $p \in M$, all of the inextendible timelike and null geodesics through $p$ are b-incomplete in the past direction.

\subsection{Stable Causality}

Spacetime $(M, g)$ is \emph{temporally orientable} if and only if a continuous vector field $F$ can be defined over $M$ where every element of $F$ is timelike with respect to $g$. When $(M,g)$ is temporally orientable, there is a sense in which the direction of time can be consistently defined at every point. However, this sense is not strong enough to capture the intuitive notion of a beginning. To see this, suppose that spacetime is everywhere b-incomplete in the direction Alice regards as the past. In that sense, Alice may be tempted to say that spacetime has a beginning. However, suppose that the direction Alice regards as the past is the direction Bob regards as the future. Instead of saying that spacetime has a beginning, Bob is tempted to say that spacetime has an ending. If we cannot objectively distinguish beginnings from endings, then spacetime does not have a beginning in the intuitive sense.

For $(M,g)$ to have a beginning, any two observers should agree on a past direction. $(M, g)$ is said to admit a \emph{cosmic time function} just in case there exists a smooth scalar field $t: M\rightarrow \mathbb{R}$ whose gradient $\nabla^\mu t$ is everywhere timelike and future-directed \citep[198-199]{HawkingEllis:1973}. We can think of $t$ as assigning a time to every point in $M$ that increases along every non-spacelike curve. This condition is equivalent to \emph{stable causality} \citep{Hawking:1969}, that is, closed timelike curves do not appear even if every light cone is opened by a small amount \citep[17]{Manchak:2020}. This prohibition is intuitively satisfying: if any two observers agree on a past direction, then no observer should be traveling backwards through time relative to another observer, and, hence, there should be no closed time-like curves. Thus, there can be a cosmic beginning only if spacetime is stably causal.  However, an objective past/future distinction requires more than spacetime admitting a cosmic time function. In addition, there must be some feature of the universe that determines the objective direction of time, e.g., a cosmic entropy gradient.

In sum, there are two necessary conditions for a spacetime to have a beginning in an intuitive sense: spacetime must be stably causal and everywhere past b-incomplete. I do not claim that the two conditions are sufficient. However, the two conditions are satisfied by models of the universe that begin with a Big Bang. Not all classic models of the Big Bang include a Big Bang \emph{singularity}. But those that do are both stably causal and b-incomplete. Moreover, if sense can be made of a cosmic beginning in some future successor to General Relativity, some similar pair of conditions will likely need to be satisfied. In the rest of this paper, I am concerned with whether spacetimes with a beginning are observationally indistinguishable from beginningless spacetimes. In the next section, I turn to discussing the notion of observational indistinguishability.

\section{Observationally Indistinguishable Spacetimes \label{obs-ind-spacetimes}}

Contemporary discussion of observationally indistinguishable spacetimes can traced back atleast to \cite{Glymour:1977} and \cite{Malament:1977}.\footnote{Other early works include \cite{Marder:1964, Ellis:1971, Glymour:1972}.} Roughly, two spacetimes are observationally indistinguishable when they ``look'' the same. Before making this notion more precise, we need to establish a few preliminaries. Two spacetimes $(M, g)$ and $(M', g')$ are \emph{isometric} just in case there exists a diffeomorphism $\phi:M\rightarrow M'$ such that $\phi_*(g) = g'$. $(M, g)$ and $(M', g')$ are \emph{locally isometric} if, for each point $p \in M$, there is an open neighborhood $O$ of $p$ which is isometric to some open neighborhood $O' \subset M'$, and vice versa. A property $P$ is \emph{local} if, given two locally isometric spacetimes $(M, g)$ and $(M', g')$,  $(M, g)$ satisfies P if and only if $(M', g')$ does so. Intuitively, whether a local property holds depends only on the geometry in arbitrarily small neighborhoods of each point.

Glymour and Malament respectively utilize notions I will call \emph{strong} and \emph{weak} observational indistinguishability. Two spacetimes $(M, g)$ and $(M', g')$ are \emph{strongly} observationally indistinguishable just in case, for every point $p\in M$, there exists a point $p' \in M'$ such that $I^-(p)$ is isometric to $I^-(p')$. When $(M, g)$ and $(M', g')$ are strongly observationally indistinguishable, any observation made in $(M, g)$ is made somewhere in $(M', g')$, and vice versa. No matter how long an observer inhabiting either spacetime waits, they will never gather data from which they can infer which of the two spacetimes they inhabit. Following \cite{Malament:1977}, spacetime $(M,g)$ is \emph{weakly} observationally indistinguishable from $(M',g')$ just in case, for every point $p \in M$, there exists a point $p' \in M'$, such that $I^-(p)$ is isometric to $I^-(p')$. Unlike strong observational indistinguishability, weak observational indistinguishability is asymmetric: the fact that all of the past light cones from $(M,g)$ are also included in $(M', g')$ does not entail that all of the past light cones from $(M', g')$ are included in $(M,g)$.  A third notion can be considered as well: $(M,g)$ is \emph{super weakly} observationally indistinguishable from $(M',g')$ if and only if there exists a region in $(M,g)$ weakly observationally indistinguishable from some region in $(M',g')$. These conditions form a hierarchy, i.e., strong implies weak, and weak implies super weak, but not vice versa.

Concerning the beginning of the universe, one might suppose that super weak observational indistinguishability would suffice. Unless humans or their successors live forever, all of the astronomical observations they ever gather fit within a finite spacetime region. However, if $(M,g)$ is merely super weakly observationally indistinguishable from $(M',g')$, then some observer in $(M,g)$ may be able to determine that they do not inhabit $(M',g')$. More impressive results are obtained if it can be shown that no observer in $(M,g)$ can infer whether they inhabit $(M',g')$, as happens with weak observational indistinguishability. For that reason, I will follow Malament (\citeyear{Malament:1977}), Manchak (\citeyear{Manchak:2009, Manchak:2011, Manchak:2016b}), and others, in considering weak observational indistinguishability.

\section{The Malament-Manchak Theorems\label{mmt-section}}

Having defined several notions of observational indistinguishability, I turn to considering the two Malament-Manchak theorems. I also sketch their proofs, because we will need the corresponding constructions later. Before sketching Manchak's proof, a brief comment on methodology. Manchak's proof follows a long-standing research tradition, where physically unrealistic mathematical constructions are used to establish physically significant theorems. As Robert Geroch wrote, unrealistic models can provide ``a guide in checking and discovering theorems'' \cite[72]{Geroch:1971}, so that ``the mere existence of a [model] having certain global features suggests that there are many models – some perhaps quite reasonable physically – with very similar properties'' \cite[78]{Geroch:1971}. Likewise, for \cite[215]{GerochHorowitz:1979}, unrealistic models allow us ``to gain an understanding of what will work and what will not, of what needs to be and can be ruled out.'' Furthermore, the cut-and-paste methodology that -- as we will see -- is used in Manchak's constructions has a good pedigree: the method is used by Stephen Hawking and George Ellis in their classic \citeyear{HawkingEllis:1973} work on global spacetime structure. In this paper, I follow the same methodology in proving various theorems, so that, while the constructions are not intended to be physically realistic, they are taken to have important physical implications.

Given a spacetime $(M, g)$ with a point $p$, the chronological past of $p$ is denoted $I^-(p)$. $(M, g)$ is \emph{bizarre} iff there exists a point $p \in M$ such that $I^-(p) = M$. Malament \citeyear{Malament:1977} conjectured that every non-bizarre spacetime is weakly observationally indistinguishable from at least one other spacetime. Manchak managed to prove two slightly stronger results:

\begin{figure}
  \centering
  \input{Clothesline}
  \caption{The clothesline construction used in proving theorem \ref{MM1}. Diagram based on \cite[41]{Manchak:2020}.}\label{clothesline}
\end{figure}

\FloatBarrier

\begin{theorem}\label{MM1}
    Let $(M, g)$ be any non-bizarre spacetime satisfying any set $C$ of local properties. Then there exists $(M', g')$, not isometric to $(M, g)$, such that (i) $(M, g)$ satisfies $C$ and (ii) $(M, g)$ is weakly observationally indistinguishable from $(M', g')$. (Established in \citep{Manchak:2009}.)
\end{theorem}

\begin{proof}
    Let's sketch the proof; for details, see \citep{Manchak:2009}. The construction appears in figure \ref{clothesline}. Start with a countable infinity of copies of $(M, g)$; label these $M(1,\alpha)$, $M(2, \alpha)$, $M(3, \alpha)$, ..., $M(i, \alpha)$, ... This collection of copies of $(M,g)$ comprises the top row in figure \ref{clothesline}. Consider a series of past light cones $C_1$, $C_2$, ..., $C_i$, ... that completely cover $(M, g)$. Identify the corresponding light cones in the $M(i, \alpha)$'s, with one light cone identified for each $i$. Next, consider another countably infinite sequence of copies of $(M, g)$ and label these $M(1, \beta)$, $M(2, \beta)$, ..., $M(i, \beta)$, ... This sequence is the bottom row in figure \ref{clothesline}. Consider a closed spacelike surface $K_1^+$ in $M(1, \alpha)$ such that $K_1^+ \not\subset C_1$ and a corresponding closed spacelike surface $K_1^+$ in $M(1, \alpha)$. Identify the bottom of $K_1^+$ in $M(1, \beta)$ with the top of $K_1^+$ in $M(1, \alpha)$ and vice versa (excluding end points). Likewise, consider a closed spacelike surface $K_2^-$ in $M(2, \alpha)$ such that $K_2^- \not\subset C_2$ and a corresponding closed spacelike surface $K_2^-$ in $M(1, \alpha)$ and vice versa (again excluding end points). Identify the top of $K_2^-$ in $M(1, \beta)$ with the bottom of $K_2^-$ in $M(2, \alpha)$ and vice versa. Repeat this procedure for all of the $M(i, \alpha)$ and $M(i, \beta)$, so that the entire sequence forms one spacetime, with $M(i, \beta)$ interleaved between $M(i, \alpha)$ and $M(i+1, \alpha)$. The resulting construction is called the \emph{clothesline} of $(M, g)$. Since the clothesline of $(M, g)$ includes duplicates of a collection of light cones that completely covers $(M, g)$, $(M, g)$ is weakly observationally indistinguishable from its clothesline.
\end{proof}

\begin{theorem}\label{MM2}
    Let $(M,g)$ be any non-bizarre spacetime satisfying any set $C$ of local properties and global conditions $G$, where $G = \{$ inextendibility, isotropy, global hyperbolicity, hole free-ness $\}$. Then there exists $(M', g')$, not isometric to $(M, g)$, such that (i) $(M, g)$ satisfies $C$, (ii) $(M, g)$ is weakly observationally indistinguishable from $(M', g')$, and (iii) $(M', g')$ does not satisfy any of the conditions in $G$. (Established in \citep{Manchak:2011}.)
\end{theorem}

\begin{proof}
    Again, I only offer a sketch;  details are in \citep{Manchak:2011}. Begin by constructing the \emph{clothesline} of $(M, g)$. In $M(1, \beta)$, remove one point that is not in $K_1^+$ or $K_2^-$. The resulting spacetime does not satisfy any of the conditions in $G$, but $(M, g)$ is weakly observationally indistinguishable from it.
\end{proof}

\noindent\textbf{Discussion.} The two Malament-Manchak theorems show that any non-bizarre spacetime, satisfying any set of local properties, is weakly observationally indistinguishable from some other spacetime with a distinct global structure, which satisfy the same set of local properties. Already this has deep cosmological implications: no matter how much data one has, and no matter how widely one gathers data, one cannot infer some of spacetime's global properties (at least if spacetime is non-bizarre). Whether we can know that spacetime began is not immediately endangered. Nonetheless, using similar methods, \cite{Manchak:2016b} established the following result:

\begin{theorem}
    Any spacetime $(M, g)$ is weakly observationally indistinguishable from a spacetime $(M', g')$ that does not admit a cosmic time function.\label{stably-causal}
\end{theorem}

\noindent Given theorem \ref{stably-causal}, that a spacetime admits a cosmic time if and only if the spacetime is stably causal, and that any spacetime with a beginning must be stably causal, any spacetime with a beginning is weakly observationally indistinguishable from one without a beginning. Several other relevant results can be established, to which I turn in the next section.

\section{Corollaries to the Malament-Manchak theorems\label{corollaries-section}}

In this section, I prove two corollaries to the Malament-Manchak theorems with skeptical implications concerning our ability to know that there is a beginning for all of spacetime. The first concerns whether we know that spacetime satisfies the application conditions for the various singularity theorems. Some non-singular spacetimes very nearly satisfy one or more of the classic singularity theorems, and only fail to satisfy those theorems because they do not satisfy a fairly specific global condition. For example, as \citet[7]{Senovilla:2021} points out, though de Sitter spacetime has no incomplete timelike geodesics, de Sitter spacetime fails to satisfy one of the classic singularity theorems only because the Cauchy surfaces in de Sitter spacetime are compact.

Following Senovilla's lead, we can investigate whether spacetimes that satisfy various global conditions required for the classic singularity theorems are weakly observationally indistinguishable from spacetimes that do not satisfy those global conditions. The classic singularity theorems assume global hyperbolicity \citep{Hawking:1965, Hawking:1966},\footnote{A set $N \subset M$ is \emph{achronal} just in case there are no two points $p, q \in N$, such that $p$ and $q$ can be connected by a timelike curve. A closed, achronal set $N \subset M$ is a \emph{Cauchy surface} just in case $N$'s domain of dependence is all of $M$. $(M, g)$ is \emph{globally hyperbolic} just in case $(M, g)$ includes a Cauchy surface.} time orientability \citep{Hawking:1965, Hawking:1966, Geroch:1966}, or the absence of closed timelike curves \citep{Hawking_Penrose:1970}. Denote this set of properties $\mathbb{S}$. Moreover, while there are generalizations of the classic theorems to semi-classical quantum gravity or using modified or averaged energy conditions, such theorems typically assume one or more of the conditions in $\mathbb{S}$, such as global hyperbolicity \citep{Tipler:1978, FewsterGalloway:2011, Wall:2013, Brown_etal:2018, FewsterKontou:2020, Fewster:2022}. Crucially, unless we know that spacetime satisfies one of the conditions in $\mathbb{S}$, we cannot know whether any of the aforementioned singularity theorems apply. Note, too, that spacetimes which fail to be time orientable or that include closed timelike curves are not stably causal. Hence, a spacetime that does not satisfy two of the conditions in $\mathbb{S}$ cannot have a cosmic beginning either.

In what follows, I first prove that any non-bizarre spacetime satisfying at least one of the members of $\mathbb{S}$ is weakly observationally indistinguishable from some spacetime does not satisfy that member. Second, I prove that any non-bizarre spacetime that is everywhere past b-incomplete weakly observationally indistinguishable from some spacetime that is not everywhere past b-incomplete.

\begin{theorem}\label{glob-cons-thm}
    Any non-bizarre spacetime $(M, g)$ satisfying one of the members of $\mathbb{S}$ is weakly observationally indistinguishable from some other spacetime that does not satisfy that member of $\mathbb{S}$.
\end{theorem}

\begin{proof} 
    I proceed with a proof by cases, where each case represents one of the three members of $\mathbb{S}$.\footnote{Thanks to J.B. Manchak for suggesting this proof in personal correspondence.} For global hyperbolicity, the result follows immediately  from Theorem \ref{MM2}. For no closed timelike curves, the result follows immediately from Theorem \ref{stably-causal}. 
    
    The proof for time orientability utilizes a clothesline construction, where the $M(1, \beta)$ region in Figure \ref{clothesline} has been modified as illustrated in Figure \ref{glob-construction}. Begin with the spacetime $(M, g)$ and construct its clothesline $(M', g')$. Assume that $O_1 \subset M(1, \beta)$ and $O_2 \subset M(1, \beta)$ are finite regions completely enclosing $K_1$ and $K_2$ respectively (as in the clothesline constructed in \citep{Manchak:2009}). Consider any finite open set $O$ in the $M(1,\beta)$ portion lying outside $O_1 \cup O_2$. By Lemma 1 of \citep{Manchak:2016b}, there exists a spacetime $(M', g'')$ where, for some open set $R \subset O$, $g''$ is flat on $R$ and $g'' = g'$ outside $O$. In this sense, a region in $(M', g')$ may be replaced by a region of Minkowski spacetime without altering $(M', g')$ outside of $O$. Next, consider Minkowski spacetime with standard $(t,x,y,z)$ coordinates and restrict attention to the rectangular block $[0,1]^4$. Identify points in two of the timelike boundaries using the isometry $\Phi: (t,0,0,0) \mapsto (1-t,1,0,0)$. Denote the spacetime that results $(M''', g''')$. By construction, $(M''', g''')$ is not time-orientable. Lastly, we can use the fact that $(M'', g'')$ and the region $R$ in $(M'', g'')$ are flat to bridge them together.  Following a similar procedure to the construction in the proof of Proposition 1 in \citep[1057]{Manchak:2016b}, let $S$ be a closed three-dimensional spacelike surface in $M'''$ and let $S'$ be a closed spacelike three-dimensional surface in $R$. $S$ and $S'$ are constructed so that they have the same size in standard Minkowski coordinates. Identify the ``bottom'' edge of $S$ with the ``top'' edge of $S'$ and the ``top'' edge of $S$ with the ``bottom'' edge of $S'$ (excluding boundary points). The resulting spacetime is not time-orientable. Since the spacetime that results includes a collection of past light cones isometric to a collection that completely covers $(M, g)$, this suffices to complete the proof.
\end{proof}

\begin{figure}
  %\centering
  \input{Theorem_3_Diagram}
  \caption{A portion of the construction used in proving Theorem \ref{glob-cons-thm}.}\label{glob-construction}
\end{figure}

\FloatBarrier

\begin{theorem}\label{past-inc-thm}
    Any non-bizarre spacetime $(M, g)$ that is everywhere past b-incomplete is weakly observationally indistinguishable from some other spacetime $(M', g')$ that is not everywhere past b-incomplete.
\end{theorem}

\begin{proof}
    The proof for theorem \ref{past-inc-thm} uses another modified clothesline, where the $M(1, \beta)$ region in Figure \ref{clothesline} has been modified as in Figure \ref{thm_5_const}. Begin with a spacetime $(M, g)$ that is everywhere past b-incomplete and construct the corresponding clothesline.\footnote{Thanks to J.B. Manchak for suggesting this proof in personal correspondence.} Once more, suppose that $O$ is some open set in the $M(1,\beta)$ portion lying outside $O_1 \cup O_2$. Pick a point $p \in O$. Once more, using lemma 1 from \citep{Manchak:2016b}, we can replace a region $R \subset O$, that includes $p$, with a flat region. Using a similar procedure as in the proof for Theorem \ref{glob-cons-thm}, identify a closed, three-dimensional spacelike surface $S$ in $R$ with a corresponding spacelike surface $S'$ in a copy of Minkowski spacetime. Identify the top edge of $S$ with the bottom edge of $S'$ and the top edge of $S'$ with the bottom edge of $S$ (excluding boundary points). Note that the Minkowski region includes complete geodesics.
\end{proof}

\begin{figure}
  %\centering
  \input{Theorem_4_Diagram}
  \caption{A portion of the construction used in proving Theorem \ref{past-inc-thm}.}\label{thm_5_const}
\end{figure}

\FloatBarrier

\noindent\textbf{Discussion.} Let's take stock. I've shown that all non-bizarre spacetimes satisfying the global conditions required for several singularity theorems are weakly observationally indistinguishable from spacetimes that do not satisfy the same global conditions. This suggests a skeptical consequence: that no matter how much data we collect, we will likely never be able to infer that the spacetime we inhabit satisfies the application conditions for the various singularity theorems. I've also shown that all non-bizarre spacetimes that are everywhere past b-incomplete are weakly observationally indistinguishable from spacetimes that are not everywhere past b-incomplete. This suggests that even if we knew that there is a singularity in our past, we would not be able to infer that a singularity precedes all of spacetime. However, a worry immediately appears.

Suppose $A$ is weakly observationally indistinguishable from $B$; in that case, $B$ is said to be a nemesis for $A$. In the two Malament-Manchak theorems, arbitrary local properties were preserved between a given spacetime and its nemesis. This is not so for theorems \ref{stably-causal}, \ref{glob-cons-thm},  or \ref{past-inc-thm}. Since local criteria, e.g., the Einstein Field Equation (EFE), various energy conditions, etc, are often used to distinguish physically reasonable spacetimes from spacetimes that are not physically reasonable, one may worry that while, e.g., b-incomplete spacetimes have nemeses, none of the nemeses are physically reasonable. In the next section, I present two constructions that help alleviate this objection. While some nemesi might violate, e.g., the Einstein Field Equations, other nemesi -- explicitly constructed in the next section -- do not. Although I do not claim that the constructed nemeses are physically reasonable, they demonstrate the difficulty in specifying the conditions to be used in ruling nemeses out.

\section{Two Nemesis Spacetimes\label{nemesis-section}}

In this section, I construct two spacetimes that help to alleviate the objection that none of the nemeses are physically reasonable. The two nemeses satisfy a variety of local conditions, which makes it difficult to specify a law that can be used, with an inductive inference, in response to the theorems from section \ref{corollaries-section}. To construct either nemesis, I will begin with an FLRW spacetime with dust stress-energy, construct the clothesline as before, and then perform surgery on the $M(1, \beta)$ region.

The surgery makes use of mathematical techniques frequently employed in astrophysics, but not yet in work on observationally indistinguishable spacetimes. The central idea involves ``gluing'' two solutions to the EFE along a shared boundary in such a way that the result also solves the EFE. As in the analogous problems in classical electrodynamics, one must ensure that various junction conditions are satisfied at the boundary. These conditions are known as the \emph{Israel Junction Conditions} (IJCs); the conditions and their derivation are sketched in \ref{IJC_appendix}.  When two spacetime regions are joined using the IJCs, the matter-energy distribution often diverges on the boundary. This spike in the matter-energy distribution is known as a \emph{thin shell}. However, astrophysicists treat thin shells as idealized approximations of physically realistic spacetimes where the matter-energy varies smoothly across a transition layer with non-zero thickness \citep{KhakshourniaMansouri:2002, KhosraviKhakshourniaMansouri:2006, DrobovTegai:2013}, a so-called thick shell.

\begin{figure}
  \centering
  \input{BD_ClotheslineOS}
  \caption{The Clothesline-OS model. In constructing the model, one starts with an FLRW spacetime with dust stress-energy content. The $M(1, \beta)$ region is then replaced by a region from a time-reversed Oppenheimer-Snyder model, where a cloud of dust expands from a white hole. $\gamma$ is an example of a complete geodesic.}\label{OSModel}
\end{figure}

\subsection{Clothesline-OS Model}

Theorem \ref{past-inc-thm} entails that, for any FLRW spacetime that is everywhere past b-incomplete, there exists a nemesis spacetime that is not everywhere past b-incomplete. Our first nemesis spacetime -- which I will call the Clothesline-Oppenheimer-Snyder (Clothesline-OS) Model -- is a nemesis for simple models of an expanding universe; as we will see, the Clothesline-OS Model is not everywhere past b-incomplete. To construct the Clothesline-OS Model, begin with any expanding FLRW spacetime $(M, g)$ with matter-energy content $T_{\mu\nu} = \text{diag}(\rho, 0, 0, 0)$, where $\rho > 0$, that is, with positive density dust. Then construct the corresponding clothesline $(M',g')$, as discussed in the proof of theorem \ref{MM1}. Next, we perform surgery on the $M(1, \beta)$ region using the well-known Oppenheimer-Snyder (OS) model. 

In the OS model, an interior FLRW spacetime region is joined, using the IJCs, to an exterior Schwarzschild region through a timelike worldtube. A detailed discussion appears in \ref{OS_appendix}. To satisfy the IJCs, a relationship between the mass parameter $M$ in the exterior and the matter-energy density $\rho$ on the interior must be satisfied; to avoid a thin shell or matter leaking into the exterior region, the interior must have a dust matter-energy content, that is, $T_{\mu\nu} = \text{diag}(\rho, 0, 0, 0)$. Oppenheimer and Snyder studied their model as an approximation of a collapsing star. However, due to the time reversal invariance of the Einstein Field Equations, the FLRW region can be either expanding or collapsing. In the case where the FLRW region is expanding, a cloud of dust erupts out of a white hole. The Cothesline-OS Model makes use of the expanding version.

Recall that, in the clothesline $(M', g')$, the $M(1, \alpha)$ and $M(2, \alpha)$ regions are connected, via closed spacelike surfaces, $K_1$ and $K_2$, to $M(1, \beta)$. We can replace $M(1, \beta)$ by the appropriate time reversed OS model, so long as the appropriate spacelike surfaces in the FLRW region of the OS model are identified with spacelike surfaces in $M(1, \alpha)$ and $M(2, \alpha)$. Note that the FLRW region in the OS model should be a copy of the corresponding region in $(M,g)$.  This is illustrated in figure \ref{OSModel}. The result, that is, the clothesline with an OS model replacing the $M(1,\beta)$ region, is the Clothesline-OS model. The mass parameter $M$ in the Schwarzschild region of the OS model and the matter-energy in the FLRW region should be chosen so that the IJCs are satisfied without a thin shell. The clothesline-OS model is not everywhere past b-incomplete, because there are complete causal curves in the Schwarzschild portion of the OS model. This is illustrated by the curve $\gamma$ in figure \ref{OSModel}.

One can show that the clothesline-OS model satisfies a number of local properties of interest. For example, since $(M', g')$ and the OS model are both solutions to the EFE, the clothesline-OS model is a solution to the EFE. Moreover, all of the standard energy conditions are satisfied by the clothesline-OS model. For dust, the null, weak, dominant, and strong energy conditions reduce to $\rho \geq 0$. By construction, this condition is satisfied by $(M, g)$ and by $(M', g')$. Since the OS model was constructed to avoid a thin shell, the OS model also satisfies this condition. Consequently, the standard energy conditions are satisfied by the clothesline-OS model. Local energy conservation (i.e., $\nabla^\mu T_{\mu\nu} = 0$) is also satisfied, as follows from the fact that the EFE is satisfied.\footnote{In general, the contracted Bianchi identity and the EFE entail that $\nabla^\mu T_{\mu\nu} = 0$. For the OS model specifically, we can check $\nabla^\mu(\Theta (f)\rho u_\mu u_\nu) = 0$, where $\Theta$ is a step function and $f$ is a monotonic function passing through zero at $\Sigma$. Taking the covariant derivative yields $\Theta (f)\nabla^\mu (\rho u_\mu u_\nu) + \delta(f) (\nabla^\mu f)\rho u_\mu u_\nu$. The first term is zero because dust satisfies local energy conservation. The second term is zero if $\nabla^\mu f$ and $u_\mu$ are perpendicular. Since $f$ is a function that changes only in a direction perpendicular to $\Sigma$, $\nabla^\mu f$ must be perpendicular to $\Sigma$. Moreover, since any particle on $\Sigma$ follows a dust worldline, $u_\mu$ must be tangent to $\Sigma$. Ergo, $\nabla^\mu f$ and $u_\mu$ are perpendicular and local energy conservation is satisfied.} Other local properties are preserved as well; for example, the metric has Lorentz signature everywhere and each region admits standard lightcone structure. While the spacetime may be globally peculiar, there is a neighborhood of every point that ``looks'' like a region from a fairly mundane spacetime. In fact, given the status of the OS model in astrophysical modeling, we should expect that any local property astrophysicists think is physically reasonable is satisfied either by the clothesline-OS model or by a closely related model with a ``thick'' shell.

\begin{figure}
  \centering
  \input{BD_MinkowskiSch}
  \caption{The Clothesline-SM model. To construct the model, one starts with the clothesline of an expanding FLRW spacetime with dust stress-energy content. As in the Clothesline-OS model, the $M(1, \beta)$ region is replaced by a time-reversed Oppenheimer-Snyder model. The region from the time-reversed Oppenheimer-Snyder model is then bridged, by a spacelike surface $S_1$, to a region from a Schwarzschild-Minkowski model (bottom left), where a thin shell of matter collapses to form a singularity. The region inside the thin shell is bridged, by another spacelike surface $S_2$, to a region from a M{\"obius} strip spacetime (bottom right). The  M{\"obius} strip spacetime is not time orientable, so the entire model is not time orientable. All stably causal spacetimes are time orientable, so the spacetime is not stably causal either.}\label{MSModel}
\end{figure}

\subsection{Clothesline-SM Model}

Theorem \ref{glob-cons-thm} entails that any FLRW spacetime, with dust stress-energy content, has a nemesis spacetime that does not satisfy one of the members of $\mathbb{S}$. Here, I will construct a nemesis spacetime that does not satisfy any of the members of $\mathbb{S}$ but which, like the clothesline-OS model, satisfies the EFE, the various energy conditions, local energy conservation, and other local properties of interest. The model is constructed by performing surgery on the clothesline-OS model. In a close cousin of the OS model, the \emph{Schwarzschild-Minkowski (SM) model}, an exterior Schwarzschild region is joined, across a thin shell, to an interior Minkowski region. For detailed discussion, see \ref{SM_appendix}. The thin shell resides on a timelike worldtube. To satisfy the IJCs, the matter-energy density $\sigma$ in the thin shell must satisfy a specific relationship with the mass parameter $M$ in the Schwarzschild exterior.

Consider, again, the clothesline-OS model. Join the Schwarzschild region, via spacelike surface $S_1$, not including boundary, to a region from an SM model, as in figure \ref{MSModel}. To ensure that the regions can be matched across $S_1$, the mass parameter from the Schwarzschild region in the SM model should be equal to the mass parameter in the original Schwarzschild region. The density $\sigma$ of the shell in the SM region should be appropriately chosen based on the mass parameter $M$ associated with the exterior Schwarzschild region, as in \ref{SM_appendix}. In the Minkowski region of the SM model, a spacelike surface $S_2$ can be identified with a spacelike surface in a M{\"o}bius strip, similar to the construction used for theorem \ref{glob-cons-thm}. Call the resulting construction the clothesline-SM model. The clothesline-SM model is not globally hyperbolic or time orientable, and includes closed timelike curves. Furthermore, the clothesline-SM model is clearly not stably causal, even though the corresponding FLRW spacetime is. The clothesline-SM model is an exact solution of the EFE, since all of the subregions are exact solutions of the EFE, and they are joined together in a way that ensures the model satisfies the EFE as a whole. Moreover, since all of the regions in the clothesline-SM model are either vacuum or positive energy dust, all of the standard energy conditions are satisfied. The clothesline-SM model, like the clothesline-OS model, also satisfies local energy conservation. The metric has Lorentz signature throughout.

The thin shell may not be physically realistic, nor the reduced regularity requirements the thin shell implies. Nonetheless, astrophysicists usually accept the SM model as an approximation to a more realistic model with a ``thick'' shell, where a shockwave collapses to form a black hole. In many standard applications, physicists adopt similar idealizations, e.g., stellar surfaces, shock waves in fluids, and surface charge distributions in electromagnetism. Instead of regarding reduced regularity as physically unrealistic, surfaces with reduced regularity are understood to approximate a rapid transition occurring on small length or time scales. Moreover, the existence of such reduced-regularity solutions is widely understood to provide strong evidence for the existence of nearby solutions satisfying higher regularity conditions. Rejecting such models merely on the grounds that they violate a regularity condition would require abandoning many empirically successful and widely employed models.

The objection can be made to both of the constructions that they require the FLRW region be dust. The actual universe can be approximated as dust at late times, but the early universe cannot be. However, an FLRW region, with non-vanishing pressure, can be matched to a region from a generalization of Schwarzschild spacetime, called \emph{Vaidya spacetime} \citep{Fayos:1991}. The FLRW/Vaidya matching can be done in a way that satisfies the IJCs without a thin shell and that satisfies all of the standard energy conditions. With non-zero pressure, material from the FLRW interior ``leaks'' into the exterior. But this can happen only after the FLRW region begins to expand from the singularity, so that the ``bottom'' portion of the regions in Figure \ref{OSModel} and the top left of Figure \ref{MSModel} can be left unmodified. Other generalizations are possible as well. For example, an FLRW region with non-zero cosmological constant can be matched with a Kottler region \citep[491]{EllisMaartensMacCallum:2012}, i.e., with a generalization of Schwarzschild spacetime with non-zero cosmological constant.

\section{The Inductive Objection\label{inductive-results-section}}

Theorems \ref{stably-causal}, \ref{glob-cons-thm}, and \ref{past-inc-thm} suggest we cannot know whether various singularity theorems apply to the spacetime we inhabit, whether our spacetime is everywhere past b-complete, or whether our spacetime is stably causal. Let's call this set the \emph{Skeptical Implications}. Let's also refer to the argument from Theorems \ref{stably-causal}, \ref{glob-cons-thm}, and \ref{past-inc-thm} for the Skeptical Implications as the \emph{Skeptical Case}. 

One objection to the Skeptical Case involves the thought that the clothesline (and the modified clotheslines used in proving subsequent theorems) are not physically plausible. Likewise, as an anonymous referee points out, clotheslines are constructed by artificially stitching copies of simple spacetimes together. Aren't the original, unstitched spacetimes more natural? And aren't the clotheslines unreasonably artificial? Holding all else equal, the objection goes, we should prefer simpler and more natural models, and so should prefer the original spacetimes over their clotheslines. 

I agree that clothesline are implausible. But recall the position maintained by Robert Geroch (and others): unrealistic examples can be useful in proving theorems and in demonstrating which properties are physically compossible with which others. For example, although the clothesline construction is used in proving theorem \ref{MM1}, the theorem  does not say that any non-bizarre spacetime is weakly observationally indistinguishable from some clothesline; instead, theorem \ref{MM1} states that any non-bizarre spacetime is weakly observationally indistinguishable from some other globally distinct spacetime. One method for proving such a result involves constructing an example, such as the clothesline, but there are many other methods for constructing nemeses \citep{Glymour:1977, Malament:1977}. Furthermore, Gordon Belot (\citeyear[chapter 8]{Belot:2023}) has shown that, given important variations on observational indistinguishability, one can construct counterparts for FLRW spacetimes that are far less ``monstrous'' than the clothesline construction.

Since so many nemeses can be constructed, any specific nemesis, on any reasonable measure over spacetimes, will be improbable. Instead of asking whether any specific nemesis is improbable, one should ask whether the disjunction of nemeses is improbable. The fact that clothesline spacetimes may be improbable does not entail that the full disjunction of nemeses is improbable. Likewise, theorem \ref{glob-cons-thm} does not say that any non-bizarre spacetime satisfying one of the members of $\mathbb{S}$ is weakly observationally indistinguishable from some modified clothesline. Instead, the theorem states that any spacetime satisfying a member of $\mathbb{S}$ is weakly observationally indistinguishable from some spacetime that does not satisfy that member of $\mathbb{S}$. Again, the fact that a modified clothesline is improbable does not show that the disjunction of nemeses violating the members of $\mathbb{S}$ is improbable. Similar replies can be offered on behalf of theorems \ref{stably-causal} and \ref{past-inc-thm}. 

There is another important objection that can be made to the Skeptical Case: while there may be nemeses for the spacetime we inhabit, the nemeses may be ruled out using induction, thereby avoiding the Skeptical Implications. However, given various sophisticated accounts of inductive inference, it is unclear that this objection succeeds. For example, at least since Nelson Goodman (\citeyear{Goodman:1955}), philosophers of science have distinguished between projectible and non-projectible predicates. Only projectible predicates support inductive inferences. In turn, on one popular account, projectible predicates are those that appear in laws of nature \citep[74]{Goodman:1955}. Most laws of nature are expressed solely in terms of local properties. Theorems \ref{MM1} and \ref{MM2} guarantee that each spacetime and its nemesis satisfy the same local properties. Since inductive inferences are supported by some set of background laws and most laws are specified in terms of local properties, the prospects are dim for using induction to escape the skeptical consequences of Theorems \ref{MM1} or \ref{MM2}.\footnote{The only non-local laws currently known involve quantum mechanical effects, such as entanglement. Given that there is no known way to transmit superluminal signals using such effects, non-local quantum effects likely cannot be used to infer the global properties of our spacetime.} A similar point can be made without appealing to Goodman. In cases of theoretical underdetermination, the available data does not allow us to choose between two or more alternative theories. But the underdetermination suggested by theorems \ref{MM1} and \ref{MM2} are not cases of \emph{theoretical} underdetermination, since those theorems allow us to hold the background theory (e.g., General Relativity) fixed. Instead, the underdetermination more closely resembles what happens in indeterministic theories, where data from some large spacetime region, such as the entire past, does not suffice for inferring what happens in some other large region, such as the future. For example, holding fixed an indeterministic quantum theory and the current state of a radioactive nucleus, our credence that the nucleus will decay in the next second should be determined entirely by the chance entailed by the quantum theory. Gathering any amount of additional data concerning the past, while holding fixed both our background theory and the state of the nucleus, cannot increase our credence that the nucleus will decay in the next second. Consequently, induction cannot increase our credence. For much the same reason, induction seems poorly suited for responding to the skepticism suggested by theorems \ref{MM1} or \ref{MM2}.

The same cannot be said with respect to an inductive response to the Skeptical Case. The proofs of theorems \ref{stably-causal}, \ref{glob-cons-thm}, and \ref{past-inc-thm} involve a spacetime constructed by replacing a portion of a clothesline with a flat region. While one can generally replace a portion of a spacetime with a flat region, doing so does not necessarily preserve arbitrary local properties. Hence, perhaps there is an inductive inference that would evade the Skeptical Implications. Call this the \emph{Inductive Objection}. Note that the Inductive Objection does not apply to at least some of the Skeptical Implications. Since the sub-proof in Theorem \ref{glob-cons-thm} for global hyperbolicity follows from Theorem \ref{MM2}, induction is no help with respect to the singularity theorems that require global hyperbolicity. Still, not all singularity theorems require global hyperbolicity, and theorems \ref{stably-causal}, \ref{glob-cons-thm}, and \ref{past-inc-thm} may still seem vulnerable to the Inductive Objection. In reply, I've shown that FLRW spacetimes (with dust stress-energy content) are generally weakly observationally indistinguishable from counterparts that preserve various interesting local properties. Furthermore, we have good reason to believe that either the nemeses satisfy local properties astrophysicists think are reasonable or else approximate closely related spacetimes that do satisfy physically reasonable local properties.  Since the observable universe can be modeled as an FLRW spacetime, we have evidence that there is a nemesis spacetime for the actual universe satisfying such local properties. In turn, it's difficult to see how there could be a projectible predicate or law that would motivate an inductive response to the Skeptical Case.

The claim is sometimes made that, at least within cosmology, inductive inferences can be based upon either the Copernican Principle, roughly, that there are no special locations in the universe, or the Cosmological Principle, roughly, that the universe is homogeneous and isotropic. Taking these in turn, first consider the Copernican Principle. The Copernican Principle cannot be the claim that any property satisfied in the neighborhood of one point is satisfied at the neighborhood of any other. Such a strong principle is easy to empirically disconfirm. Plausibly, the Copernican Principle should be understood as the principle that the laws do not distinguish any observer, point, or region. But then given that $(M,g)$ and its nemesis share a variety of interesting local properties, it becomes difficult to say how the nemesis could violate the Copernican Principle when $(M,g)$ does not. Given two spacetimes satisfying the same laws, any symmetry of the laws in the first must also be a symmetry of the laws in the second.

Now turn to the Cosmological Principle. At first pass, one may worry about the circularity of ruling out spacetime models because they violate the Cosmological Principle. We are investigating whether an inference can be successfully made that spacetime satisfies various global properties; assuming that spacetime satisfies some global properties (e.g., the Cosmological Principle), in order to infer other global properties, seems infelicitous. Set this worry to one side. The popularity of the Cosmological Principle has waned since the mid-twentieth century. Instead of assuming general principles, cosmologists now investigate what can result from generic initial conditions, and accept general principles only when they can be shown to result from such conditions.  Writing in the late 1970s, Amal Kumar Raychaudhuri described physicists who were not happy to simply assume the Cosmological Principle. Raychaudhuri's physicists wanted

\begin{quote}
    to demonstrate that starting with arbitrary initial conditions, physical processes lead to an evolution of a universe obeying the cosmological principle – inhomogeneities and anisotropies being smoothed out in a time small compared to the age of the universe \cite[4]{Raychaudhuri:1979}.
\end{quote}

\noindent Raychaudhuri describes two strategies. First, one could try to show that arbitrary initial conditions, evolved forward, generically come to approximate an FLRW spacetime. This strategy has not been successful. Second, one could try to show that arbitrary initial conditions, when evolved forward, generically result in spacetimes with large, approximately FLRW regions. This is the strategy that inflationary cosmologists pursued in the decade after 1979. George Ellis, Roy Maartens, and Malcolm MacCallum summarize,

\begin{quote}
    [...] on any reasonable measure the FLRW models are of zero probability within the family of all possible cosmological models. Thus this assumption implies that the universe is of an extremely special, hence highly fine-tuned, nature. From the late 1930s to the 1960s, this was taken as reasonable. However, in the late 1970s, the philosophical tide turned: it became common to assume the universe has a generic rather than special geometric nature. Indeed it was no longer assumed there were any specific cosmological principles at all \cite[347]{EllisMaartensMacCallum:2012}.
\end{quote}

\noindent Given this change in cosmological methodology, it is difficult to justify a version of the Inductive Objection based on the Cosmological Principle. On the view that the visible universe has a generic rather than a special geometric nature, spacetime may have almost any global structure as long as some sufficiently large region matches the visible universe. As early as \citeyear{CollinsHawking:1973}, C.B. Collins and Stephen Hawking presented an anthropic explanation for the fact that we observe an isotropic universe. Although isotropic spacetimes comprise a set of measure zero in the space of all homogeneous spacetimes, the large scale structures required for life would not be possible if the visible universe were not close to being isotropic. As Collins and Hawking argue, most members of an ensemble of universes with arbitrary initial conditions are anisotropic, but we can only exist in the isotropic minority. Likewise, within inflationary cosmology, the visible universe is homogeneous and isotropic, but only because a small region, which may have differed almost arbitrarily from the rest of spacetime, underwent inflation. Since spacetime as a whole may differ dramatically from that small region, it's difficult to understand how the Cosmological Principle could be used to infer specific global properties. Furthermore, as a recent review article \cite[8]{COMPACT:2026uzi} notes, quantum cosmologies assign significant weight to the universe possessing a non-trivial initial topology. The authors point out that for this non-trivial topology to have observable consequences appears to require some degree of fine-tuning; for example, if the inflationary epoch is too long, then any effects due to the non-trivial topology would be located far beyond our cosmological horizon.

Of course, inflation is only one way of showing that an FLRW region can result from generic initial conditions. Whatever mechanism one postulates, if we should consider the visible universe the result of evolving from generic initial conditions, we should be hesitant to postulate any specific global properties.

\section{Conclusion\label{conclusion}}

I've presented several results suggesting that we cannot know whether spacetime had a beginning. First, any non-bizarre stably causal spacetime is weakly observationally indistinguishable from a spacetime that is not stably causal; for this reason, any spacetime that has the sort of time ordering required for a beginning is weakly observationally indistinguishable from another spacetime without that time ordering. Second, while various singularity theorems require spacetime to fulfill one or more global conditions, any non-bizarre spacetime satisfying such conditions is weakly observationally indistinguishable from some spacetime that does not. So, while we can identify conditions that entail that spacetime is singular, we likely cannot infer that our spacetime satisfies those conditions. Third, any non-bizarre spacetime that is everywhere past b-incomplete is observationally indistinguishable from some other spacetime that is not everywhere past b-incomplete. This suggests that even if we knew that our past is b-incomplete, we could not infer everywhere past b-incompleteness. Fourth, I constructed two of the nemeses for expanding FLRW spacetimes populated by positive density dust. These nemeses satisfy a number of important local properties (e.g., the EFE, the various energy conditions, local energy conservation, etc); in turn, this suggests that no induction based on laws can rule out relevant nemeses for the spacetime we inhabit. Other inductive inferences -- such as those based on the Copernican or Cosmological Principles -- do not seem promising either. Together, these results support agnosticism concerning a cosmic beginning.\\

\noindent\textbf{Acknowledgments.} Thanks to J.B. Manchak for suggesting several of the mathematical techniques used in this paper and for providing valuable critical feedback. Without his guidance, I would not have been able to write this paper. Also, thanks to Sarah Geller, Enrique Gaztañaga, Levi Greenwood, Ettore Minguzzi, Eric Poisson, James Read, J.M.M. Senovilla, and Balša Terzić; all eight provided valuable conversations or correspondence.

\appendix
\renewcommand{\thesection}{\appendixname~\Alph{section}}
\section{The Israel Junction Conditions}\label{IJC_appendix}

In this appendix, I introduce and sketch a derivation of the Israel Junction Conditions. For further detail and applications, the reader is encouraged to consult standard references \citep{Israel:1966, MisnerThorneWheeler:1973,Poisson:2004, EllisMaartensMacCallum:2012, Art:2023}. Relativistic spacetimes are often assumed to be Hausdorff, paracompact $C^\infty$ manifolds with smooth Lorentzian metrics \citep[12]{Wald:1984}, but such requirements are frequently relaxed when considering the joining of two spacetime regions \citep{ClarkeDray:1987}. To join two spacetime regions at a shared boundary, it is often necessary to generalize the metric and other quantities of interest to be valid in the distributional sense. In that context, the first derivative of the metric may be discontinuous at the boundary, leading to a delta function contribution to the curvature. Ergo, the metric must be $C^2$ on either side of the boundary,  and can be $C^0$ across the boundary.

Pursuing this thought, consider two $C^3$ manifolds $\mathcal{V}^+$ and $\mathcal{V}^-$, with boundaries $\Sigma^+ \subset \mathcal{V}^+$, $\Sigma^- \subset \mathcal{V}^-$, equipped with $C^2$ Lorentzian metrics $g^+$ and $g^-$ respectively. Identify the points in $\Sigma^+$ and $\Sigma^-$ via the $C^3$ diffeomorphism $\phi:\Sigma^+\mapsto\Sigma^-$, which is assumed to have a $C^3$ inverse. Once identified, we can define $\Sigma\equiv\Sigma^\pm$ and $\mathcal{V} \equiv \{\mathcal{V}^+ \cup \Sigma \cup \mathcal{V}^-\}$. For the sake of convenience, define a set of coordinates $x^\mu$ over $\mathcal{V}$. The coordinate dependence is assumed for the sake of convenience and will later be dropped. Define a scalar field $l(x^\mu)$ on $\mathcal{V}$. Impose three restrictions on $l(x^\mu)$: (i) $l(x^\mu) < 0$ for $x^{\mu} \in \mathcal{V}^-$, (ii) $l(x^\mu) = 0$ for $x^\alpha \in \Sigma$, and (iii) $l(x^\mu) > 0$ for $x^{\mu} \in \mathcal{V}^+$. Then, for the metric $g$ defined over all of $\mathcal{V}$, we have:

\begin{equation}
    g_{\alpha\beta} = \Theta(l)g^+_{\alpha\beta}+\Theta(-l)g^-_{\alpha\beta}
\end{equation}

\noindent where $\Theta$ is the Heaviside function, i.e., $\Theta(l) = 1$ for $l > 0$, $\Theta(l) = 0$ for $l < 0$, and $\Theta(l = 0)$ is undefined. The following notation will be useful as well. For the distributional Einstein Field Equation (dEFE) to be satisfied, the Christoffel symbol will need to be distributionally valid. The Christoffel symbol $\Gamma^\alpha_{\beta\gamma}$ is defined as:

\begin{equation}
    \Gamma^\alpha_{\beta\gamma} = \frac{1}{2}g^{\alpha\lambda}\left(g_{\lambda\beta,\gamma}+g_{\lambda\gamma,\beta}-g_{\beta\gamma,\lambda}\right) 
\end{equation}

\noindent where commas indicate distributional partial derivatives. Evaluating the distributional partial derivatives, we find:

\begin{equation}\label{derivative}
    \partial_\lambda g_{\alpha\beta} = \Theta(l)g^+_{\alpha\beta,\lambda}+\Theta(-l)g^-_{\alpha\beta,\lambda} +  \left(g^+_{\alpha\beta}-g^-_{\alpha\beta}\right)\delta(l)\partial_\lambda l
\end{equation}

\noindent Thus, the Christoffel symbol will generally include terms involving the product of a Heaviside function and a $\delta$-function. Since the product of a Heaviside function and a $\delta$-function is not well-defined in the distributional sense, the term involving $\delta(l)$ in \ref{derivative} must vanish. Hence, the metric is continuous at $\Sigma$:

\begin{equation}\label{pre-FIJC}
    [g_{\alpha\beta}] = 0
\end{equation}

\noindent We can now relax the coordinate dependence by expressing equation \ref{pre-FIJC} in terms of quantities intrinsic to $\Sigma$. Suppose we have some coordinates $y_a$ on $\Sigma$. We can define a set of holonomic basis vectors in $\Sigma$, i.e., $e^\alpha_a \equiv \partial x^\alpha/\partial y^a$. In turn, the metric induced on $\Sigma$ can be written in terms of $h_{ab}$, the metric restricted to $\Sigma$, as $h_{ab} \equiv g_{\alpha\beta}e^\alpha_a e^\beta_b$. The left hand side of equation \ref{pre-FIJC} can be re-written as $[g_{\alpha\beta}e^\alpha_a e^\beta_b]$. Hence, the first IJC:

\begin{equation}
    [h_{ab}] = 0
\end{equation}

\noindent The first IJC is needed to ensure that the matching of $\Sigma^+$ and $\Sigma^-$ is geometrically sensible.

In the case of the ordinary EFE, so long as one considers a sufficiently smooth manifold, one can always construct a corresponding matter-energy distribution. The same is true for the dEFE; the first IJC ensured a sensible geometry, while the second provides a condition on the matter-energy distribution so that the dEFE is satisfied. The Einstein tensor is calculated from partial derivatives of the Christoffel symbol, which can be written as

\begin{equation}\label{Christ-Symbol}
    \Gamma^\alpha_{\beta\gamma,\delta} = \Theta(l)\Gamma^{\alpha+}_{\beta\gamma,\delta} + \Theta(-l)\Gamma^{\alpha-}_{\beta\gamma,\delta} + [\Gamma^{\alpha}_{\beta\gamma}]\delta(l)\partial_\delta l
\end{equation}

\noindent Though $\Gamma^\alpha_{\beta\gamma,\delta}$ involves a $\delta$-function, this will not result in the product of a Heaviside function and $\delta$-function in the Riemann tensor $R^\alpha_{\beta\gamma\delta}$. Instead, we have:

\begin{equation}\label{R-equation}
    R^\alpha_{\beta\gamma\delta} = \Theta(l)R^{\alpha+}_{\beta\gamma\delta} + \Theta(-l)R^{\alpha-}_{\beta\gamma\delta} + \delta(l)A^\alpha_{\beta\gamma\delta}
\end{equation}

\noindent where $A^\alpha_{\beta\gamma\delta} \equiv [\Gamma^\alpha_{\beta\delta}]\partial_\gamma l - [\Gamma^\alpha_{\beta\gamma}]\partial_\delta l$. The demand that there is no $\delta$-function term is equivalent to $A^\alpha_{\beta\gamma\delta} = 0$ and would avoid a singularity in the Riemann, Einstein, or stress-energy tensors. The extrinsic curvature of a hypersurface is defined as 

\begin{equation}
    K_{ab} \equiv e^\alpha_a e^\beta_b \nabla_\alpha n_\beta
\end{equation}

\noindent In the case where $A^\alpha_{\beta\gamma\delta} = 0$, one can show that $[K_{ab}]=0$, so that the extrinsic curvature of $\Sigma$ is the same when considered from either side.

However, there is a well-known physical interpretation for the case where $A^\alpha_{\beta\gamma\delta} \neq 0$. Using \ref{Christ-Symbol} and \ref{R-equation}, one obtains an expression for the Einstein Tensor:

\begin{equation}\label{Einstein-Tensor}
    G_{\alpha\beta} = \Theta(l)G^+_{\alpha\beta} + \Theta(-l)G^-_{\alpha\beta} + \delta(l)E_{\alpha\beta}
\end{equation}

\noindent where $E_{\alpha\beta} \equiv A_{\alpha\beta}-\frac{1}{2}Ag_{\alpha\beta}$.  When $A^\alpha_{\beta\gamma\delta}$ is non-vanishing, there is a $\delta$-function in the Einstein tensor. There is a corresponding $\delta$-function term in the stress energy tensor, corresponding to a thin shell of matter-energy on $\Sigma$. Using \ref{Einstein-Tensor} and the dEFE, one can show that the extrinsic curvature on $\Sigma$ is discontinuous: 

\begin{equation}
    [K_{ab}] = -8\pi\left(S_{ab} - \frac{1}{2}S h_{ab}\right)
\end{equation}

\noindent where $S_{ab}$ is the surface stress-energy on $\Sigma$ and $S \equiv h^{ab}S_{ab}$ its trace. Equivalently, one can show

\begin{equation}
    S_{ab} = -\frac{1}{8\pi}\left([K_{ab}] - \frac{1}{2} [K]\, h_{ab}\right)
\end{equation}

\noindent This is the second IJC, i.e., that either $[K_{ab}] = 0$ or else the discontinuity in $K_{ab}$ is compensated for by stress-energy on $\Sigma$. After deriving the thin shell required to compensate for a given discontinuity in the extrinsic curvature, one can check whether $S_{ab}$ is consistent with, e.g., the various energy conditions.

\section{The Oppenheimer-Snyder Model}\label{OS_appendix}

The Oppenheimer-Snyder model represents a uniform and spherically symmetric dust cloud, with vacuum exterior, that collapses to form a black hole. By Birkhoff's theorem, the exterior to the cloud admits a Schwarzschild metric, with line element

\begin{equation}
    (ds^2)^+ = -FdT^2 + F^{-1}dR^2 + R^2d\Omega^2
\end{equation}

\noindent where $F \equiv 1-2M/R$ and $d\Omega^2 \equiv d\theta^2 + \sin^2\theta d\phi^2$. The collapsing dust cloud is modeled as a region admitting the FLRW line element

\begin{equation}\label{FLRW_line_element}
    (ds^2)_- = -dt^2 + a(t)^2\left(d\chi^2 + S_k(\chi)^2\,d\Omega^2\right)
\end{equation}

\noindent where $a(t)$ is the scale factor, $\chi$ is a co-moving ``radial'' coordinate, and

\begin{equation}
    S_k(\chi) \equiv \begin{cases}\sin\chi & (k=+1)\\ \chi & (k=0)\\ \sinh\chi &
    (k=-1)\end{cases}, \qquad C_k(\chi)\equiv \frac{dS_k}{d\chi} =
    \begin{cases}\cos\chi\\ 1\\ \cosh\chi\end{cases}
\end{equation}

\noindent The two regions are matched through a timelike surface $\Sigma$. On constant-time surfaces, $\Sigma$ is a sphere. The stress-energy content is dust, i.e., $T_{\mu\nu}=\text{diag}(\rho, 0, 0, 0)$. Define $\tau$ as the time kept by a clock that remains on $\Sigma$.

The first IJC is fulfilled because the induced metric can be matched from both sides. An embedding for $\Sigma$ can be defined in terms of the interior co-moving coordinate $\chi=\chi_0$ (constant), i.e., $x^-_\alpha(\tau, \theta, \phi) = (t=\tau, \chi=\chi_0, \theta=\theta, \phi=\phi)$. Likewise, there is an embedding for $\Sigma$ in terms of exterior coordinates, i.e., $x^+_\alpha(t, \theta, \phi) = ( T = T(\tau), R = R(\tau), \theta = \theta, \phi = \phi)$. The Jacobian relating the interior or exterior coordinates to the embedding coordinates is $(e^\alpha_a)^\pm = (\partial x^\alpha/\partial y^a)^\pm$. For the interior, $(e^\alpha_a)^- = \text{diag}(1, 0, 1, 1)$, while, for the exterior,

\begin{equation}
    (e^\alpha_a)^+ = \begin{pmatrix}
                    \dot{T} & 0 & 0 \\
                    \dot{R} & 0 & 0 \\
                    0 & 1 & 0 \\
                    0 & 0 & 1 \\
                \end{pmatrix}
\end{equation}

\noindent The induced metric $h_{ab}$ is defined as $h^\pm_{ab} = g^\pm_{\alpha\beta}e^\alpha_a e^\beta_b$. For the interior, $h_{ab}^- = \text{diag}(-1, a^2(t)S_k^2(\chi_0), a^2(t)S_k^2(\chi_0)\sin^2\theta)$, while, for the exterior, $h_{ab}^+ = \text{diag}(-F\dot{T}^2+F^{-1}\dot{R}^2, R^2, R^2\sin^2\theta)$. Matching the interior induced metric with the exterior yields two conditions. First, $-1 =-F\dot{T}^2 + F^{-1}\dot{R}^2$, and, second, $R = a(t)S_k(\chi_0)$.

For the second IJC, the extrinsic curvature $K^\pm_{ab}$ can also be computed using the Jacobian, i.e., $K_{ab} = e^\alpha_a e^\beta_b \nabla_\alpha n_\beta$, where $n_\beta$ is a vector normal to $\Sigma$. To compute $n_\beta$ in terms of the interior coordinates, it is convenient to define a scalar field $\ell$ which monotonically increases from the interior to the exterior of $\Sigma$. A good choice is $\ell = \chi - \chi_0$. The normal vector can then be computed using $n_\mu = \partial_\mu \ell / \sqrt{|g^{\alpha\beta}(\partial_\alpha \ell)(\partial_\beta \ell)|}$, where the only non-zero component is $n^-_r = a(\tau)$. This yields $K_{ab}^- = \text{diag}(0, a(\tau)S_k(\chi_0)C_k(\chi_0), K_{\theta\theta}^-\sin^2\theta)$.

Next, compute $K_{ab}^+$. Due to spherical symmetry, $n^+_\mu$ has, at most, $T$ and $R$ components, i.e., $n_\mu^+ = (n_T, n_R, 0, 0)$. An orthogonality condition, i.e., $n_\mu e^\mu_\tau = 0$, implies $n_\mu^+ (e^\mu_\tau)^+ =  -\frac{\dot{R}}{\dot{T}}n_R$. Demanding that $n_\mu$ is normalized entails $n_R^2 \left( -F^{-1}\dot{R}^2 + F\dot{T}^2 \right) = \dot{T}^2$. The trajectory of any particle on $\Sigma$ is timelike, with four-velocity $u^\mu = (\dot{T}, \dot{R}, 0, 0)$. Since $u^\mu$ is timelike, $g_{\mu\nu}u^\mu u^\nu = -1 = -F\dot{T}^2+F^{-1}\dot{R}^2$. Thus, $n_R = \dot{T}$ and, consequently, $n_T = -\dot{R}$. Therefore, $n_\mu^+ = (-\dot{R}, \dot{T}, 0, 0)$. One can now show $K^+_{\theta\theta} = RF\dot{T}$ and $K^+_{\phi\phi} = K^+_{\theta\theta}\sin^2\theta$. Additional care must be taken with respect to the $\tau\tau$ component of $K^+_{ab}$, which can be written $K^+_{\tau\tau} = e^\mu_\tau e^\nu_\tau \nabla_\mu n_\nu$. The four-velocity $u^\mu$ of a particle sitting on $\Sigma$ can be defined such that $(e^\mu_\tau)^+ = u^\mu$. Using this definition, one can write $K^+_{\tau\tau} = u^\mu u^\nu \nabla_\mu n_\nu$. A standard identity then entails $K^+_{\tau\tau} = -n_\nu u^\mu \nabla_\mu u^\nu$. Since $u^\mu \nabla_\mu u^\nu$ is the four-acceleration $a^\nu$ of $\Sigma$, $K^+_{\tau\tau} = -n_\nu a^\nu$. Hence, $K_{ab}^+ = \text{diag}(-n_\nu a^\nu, RF\dot{T}, K_{\theta\theta}^+\sin^2\theta)$.

To match the extrinsic curvature on both sides, two conditions must be satisfied. First, $-n_\nu a^\nu = 0$ and, second, $a(\tau) S_k(\chi_0)C_k(\chi_0) = RF\dot{T}$. The second condition, in conjunction with $R = a(\tau)S_k(\chi_0)$, yields an equation of motion for $\Sigma$. The first condition implies that the four-acceleration of a particle on $\Sigma$ must be zero. This is physically intuitive. Since the interior is dust, there is no pressure from the interior on any of the particles on $\Sigma$. Likewise, the exterior is vacuum, so there is no exterior pressure. With no applied force, $\Sigma$ traverses a geodesic trajectory. Therefore, the four-acceleration is zero, i.e., $u^\mu\nabla_\mu u^\nu = 0$, and $[K_{\tau\tau}]=0$.

Lastly, an intuitive interpretation can be given for the parameter $M$. Combining $\dot{T} = F^{-1}C_k(\chi_0)$, $\dot{R}=\dot{a}S_k(\chi_0)$, $-1 = -f\dot{T}^2+F^{-1}\dot{R}^2$, and $C_k^2+kS_k^2=1$  yields $F = 1-(k+\dot{a}^2)kS_k(\chi_0)^2$. Using $F = 1-2M/R$ and $R = a(\tau) S_k(\chi_0)$, one finds $M = \frac{1}{2}a(\dot{a}^2+k)S_k(\chi_0)^3$. Substituting one of the FLRW equations, $M = \frac{4\pi}{3} a(\tau)^3 S_k(\chi_0)^3 \rho$. Thus, $M$ is the mass of the material in an FLRW region with density $\rho$; this is the \emph{Misner-Sharp mass}.

\section{Schwarzschild \& Minkowski}\label{SM_appendix}

A Minkowski interior can be matched to a Schwarzschild exterior through a thin shell. The shell either expands or collapses; if the shell collapses, a black hole forms. The Minkowski interior admits coordinates with the following line element:

\begin{equation}
    (ds^2)^- = -dt^2 + dr^2 + r^2d\Omega^2
\end{equation}

\noindent where $d\Omega^2 \equiv d\theta^2 + \sin^2\theta d\phi^2$, and a Schwarzschild exterior, with line element:

\begin{equation}
    (ds^2)^+ = -FdT^2 + F^{-1}dR^2 + R^2d\Omega^2
\end{equation}

\noindent The two regions are matched through a timelike surface $\Sigma$. On constant-time surfaces, $\Sigma$ is a sphere.

The first IJC states that the induced metric $h_{ab}$ is continuous between both sides of $\Sigma$. To calculate the induced metric in terms of interior coordinates, we can define an embedding, i.e., $x^-_\alpha(t, \theta, \phi) = ( t = t(\tau), r = r(\tau), \theta = \theta, \phi = \phi)$, with a corresponding Jacobian $e^\alpha_a = \partial x^\alpha/\partial y^a$. Using $e^\alpha_a$, one can compute $h^-_{ab} = \text{diag}(-\dot{t}^2 +\dot{r}^2, r^2, r^2 \sin^2\theta)$. For the exterior, we have the embedding $x^+_\alpha(t, \theta, \phi) = ( T = T(\tau), R = R(\tau), \theta = \theta, \phi = \phi)$, yielding $h^+_{ab} = \text{diag}(-F\dot{T}^2 + F^{-1}\dot{R}^2, R^2, R^2\sin^2\theta)$. Matching $h^+_{ab}$ with $h^-_{ab}$ entails two kinematic conditions. First, $-\dot{t}^2+\dot{r}^2 = -F\dot{T}^2+F^{-1}\dot{R}^2$ and, second, $R=r$. Note the latter is similar to the  Oppenheimer-Snyder model.

Now turn to the second IJC. The extrinsic curvature is discontinuous between the interior and exterior regions, requiring a thin shell is needed to compensate. The extrinsic curvature computed for the Schwarzschild exterior was found in \ref{OS_appendix}. For the Minkowski interior, spherical symmetry entails $n_\mu$ has, at most, $t$ and $r$ components, $n_\mu = (n_t, n_r, 0, 0)$. The orthogonality condition $n_\mu e^\mu_\tau = 0$ entails $n_t = -\dot{r}n_r/\dot{t}$. Using this relationship and that $n_\mu n^\mu = 1$, one obtains $n_r^2 \left( \dot{t}^2 - \dot{r}^2 \right) = \dot{t}^2$. Since $\dot{t}^2-\dot{r}^2 = 1$, $n_r = \dot{t}$ and $n_t=-\dot{r}$. Hence, $n_\mu = (-\dot{r}, \dot{t}, 0, 0)$. Now, $n_\mu$ can be used to compute $K^-_{ab}$. For the angular components, $K^-_{\theta\theta} = r\dot{t}$ and $K^-_{\phi\phi} = K^-_{\theta\theta}\sin^2\theta$. Plugging $u^\mu = e^\mu_\tau$ into $K^-_{\tau\tau}$ yields $K^-_{\tau\tau} = u^\mu u^\nu \nabla_\mu n_\nu$. The covariant derivative can be moved to $u^\nu$, i.e., $K^-_{\tau\tau} = -n_\nu u^\mu \nabla_\mu u^\nu$. Because $u^\mu \nabla_\mu u^\nu$ is the four-acceleration $a^\nu$ of $\Sigma$, $K^-_{\tau\tau} = -n_\nu a^\nu$. 

If there were no thin shell, $K^+_{ab}$ would be equal to $K^-_{ab}$. From the first junction condition, we know that $r=R$. A normalization condition ($n_\mu n^\mu = 1$) entails $\dot{T} = \left(\sqrt{\dot{R}^2+F}\right)/F$ and $\dot{t} = \sqrt{\dot{r}^2+1}$. Putting these conditions together with $K^\pm_{\theta\theta}$ and enforcing $[K_{\theta\theta}]=0$ entails $1-2M/R = 1$. Hence, $M=0$. This makes intuitive sense: with a zero mass parameter, the interior region reduces to zero size, and Schwarzschild reduces to Minkowski. 

Assuming the Minkowski interior has a non-zero radius, there must be a thin shell. The shell's stress-energy $S_{ab} = -(8\pi)^{-1}([K_{ab}]-[K]h_{ab})$. $K_{ab}$ and $K$ are discontinuous, i.e., 

\begin{equation}
    [K_{ab}] = \text{diag}(-(n_\nu a^\nu)^+ + (n_\nu a^\nu)^-, R\left(\sqrt{\dot{R}^2+F} - \sqrt{\dot{R}^2+1}\right), [K_{\theta\theta}]\sin^2\theta)
\end{equation}

\begin{equation}
    [K] = -[K_{\tau\tau}] + 2R^{-2}[K_{\theta\theta}]
\end{equation}

\noindent In turn,

\begin{equation}
    S_{ab} = \text{diag}\left(-\frac{[K_{\theta\theta}]}{4\pi R}, -\frac{1}{8\pi}(-[K_{\theta\theta}] + R^2[K_{\tau\tau}]), S_{\theta\theta}\sin^2\theta\right)
\end{equation}

\noindent For the matter-energy populating spacetime to be dust, $S_{ab}$ must equal $\text{diag}(\sigma, 0, 0)$, implying two requirements:

\begin{comment}
    \begin{equation}
    \begin{aligned}
        \sigma &= -\frac{1}{4\pi R}[K_{\theta\theta}] = \frac{1}{4\pi}\left(\sqrt{\dot{R}^2+1}-\sqrt{\dot{R}^2+F}\right)\\
        0 &= -[K_{\theta\theta}]+R^2 [K_{\tau\tau}] = 4\pi R \sigma + R^2[K_{\tau\tau}]
    \end{aligned}
\end{equation}
\end{comment}

\begin{equation}
    \begin{aligned}
        \sigma &= \frac{1}{4\pi}\left(\sqrt{\dot{R}^2+1}-\sqrt{\dot{R}^2+F}\right)\\
        0 &= 4\pi R \sigma + R^2[K_{\tau\tau}]
    \end{aligned}
\end{equation}

\noindent The latter entails $[K_{\tau\tau}] = -(4\pi/R)\sigma$. This condition is satisfied. For dust, the standard energy conditions are satisfied if $\sigma > 0$. In turn, $\sigma > 0$ as long as $M$, the mass parameter from the Schwarzschild exterior, is positive. This makes intuitive sense: a hollow shell of mass has the same external gravitational field as a ``point mass'' (or singularity) at the center. This is a relativistic version of Newton's shell theorem.

A spherical shell collapses to zero size in finite time. To see this, start by defining $\beta^+ \equiv \sqrt{\dot{R}^2+F}$ and $\beta^- \equiv \sqrt{\dot{R}^2+1}$, so that $S_{\theta\theta} = \frac{1}{8\pi}\left(\beta^+-\beta^-+\frac{R}{\dot{R}}\left(\dot{\beta}^+-\dot{\beta}\right)\right)$. Since dust stress-energy was assumed, $\dot{R}(\beta^+-\beta^-) + R(\dot{\beta}^+-\dot{\beta}^-) = 0$. This is equivalent to the $\tau$-derivative of $R(\beta^+-\beta^-)$ and, hence, $R(\beta^--\beta^+) = C$, where $C$ is a constant of integration. An equation of motion follows for the shell:

\begin{equation}
    \dot{R}^2 - \frac{M^2}{C^2} - \frac{C^2}{4R^2} + \frac{M}{R} + 1 = 0
\end{equation}

\noindent At small $R$, $\dot{R}^2 \approx C^2/(4R^2)$. Since $C^2/(4R^2)$ is strictly greater than zero, $\dot{R}^2$ is strictly greater than zero. Thus, for small $R$, $\dot{R}$ is never zero and $R$ must be monotonically increasing or decreasing. If the shell is collapsing, the shell will not stop collapsing until the shell reaches zero radius. Assuming a collapsing shell, separation of variables gives $R_0^2/C \approx \Delta \tau$, where $\Delta \tau \equiv \tau_f - \tau_0$. Since $\Delta \tau$ is finite, the shell collapses to zero radius in finite proper time.

\bibliographystyle{apa}
\bibliography{references}

\end{document}

%% file: Clothesline.tex
%\documentclass[tikz,border=4pt]{standalone}
%\usepackage{tikz}
%\usetikzlibrary{arrows.meta}
%\usetikzlibrary{calc}

%\begin{document}
\begin{tikzpicture}[
  line cap=round,
  line join=round,
  line width=0.9pt
]

% ---------- panel geometry ----------
\def\W{6.2}
\def\H{3.9}
\def\Gx{0.7}
\def\Gy{0.8}

% ---------- helpers ----------
\newcommand{\openseg}[4]{%
  \draw (#1) -- (#2);
  \draw[fill=white] (#1) circle (0.08);
  \draw[fill=white] (#2) circle (0.08);
  \node[#4] at ($(#1)!0.5!(#2)$) {$#3$};
}

\newcommand{\lightcone}[3]{%
  \begin{scope}
    \clip (#2) rectangle (#3);
    \draw[dash pattern=on 2.2pt off 2.2pt]
      (#1) -- ++(225:20);
    \draw[dash pattern=on 2.2pt off 2.2pt]
      (#1) -- ++(315:20);
  \end{scope}
}

\tikzset{
  myarrow/.style={
    {Stealth[length=4mm,width=4mm]}-{Stealth[length=4mm,width=4mm]},
    line width=0.9pt
  }
}

% =========================
% Box 1 (top-left)
% =========================
\coordinate (A0) at (0,0);
\coordinate (A1) at ($(A0)+(\W,\H)$);
\draw (A0) rectangle (A1);

\coordinate (A_L) at ($(A0)+(1.0,2.35)$);
\coordinate (A_R) at ($(A0)+(2.8,2.35)$);
\openseg{A_L}{A_R}{K_1^{+}}{below=3pt}

\coordinate (p1) at ($(A0)+(4.7,3.1)$);
\fill (p1) circle (0.08);
%\node at ($(p1)+(0.5,0.12)$) {$p_1$};

\node at ($(p1)+(0,-1.5)$) {$C_1$};

\lightcone{p1}{A0}{A1}

% =========================
% Box 2 (top-right)
% =========================
\coordinate (B0) at ($(A0)+(\W+\Gx,0)$);
\coordinate (B1) at ($(B0)+(\W,\H)$);
\draw (B0) rectangle (B1);

\coordinate (B_topL) at ($(B0)+(3.0,2.75)$);
\coordinate (B_topR) at ($(B0)+(4.8,2.75)$);
\openseg{B_topL}{B_topR}{K_2^{+}}{above=6pt, xshift=-10pt}

\coordinate (B_botL) at ($(B0)+(3.0,2.05)$);
\coordinate (B_botR) at ($(B0)+(4.8,2.05)$);
\openseg{B_botL}{B_botR}{K_2^{-}}{above=0pt}

\coordinate (p2) at ($(B0)+(1.9,2.55)$);
\fill (p2) circle (0.08);
%\node at ($(p2)+(-0.5,0.25)$) {$p_2$};

\node at ($(p2)+(0,-1.3)$) {$C_2$};

\lightcone{p2}{B0}{B1}

% =========================
% Box 3 (bottom-left)
% =========================
\coordinate (C0) at ($(A0)+(0,-\H-\Gy)$);
\coordinate (C1) at ($(C0)+(\W,\H)$);
\draw (C0) rectangle (C1);

\coordinate (C_L1) at ($(C0)+(1.2,1.95)$);
\coordinate (C_R1) at ($(C0)+(3.2,1.95)$);
\openseg{C_L1}{C_R1}{K_1^{+}}{above=3pt}

\coordinate (C_L2) at ($(C0)+(4.0,1.95)$);
\coordinate (C_R2) at ($(C0)+(5.8,1.95)$);
\openseg{C_L2}{C_R2}{K_2^{-}}{below=3pt}

% =========================
% Box 4 (bottom-right)
% =========================
\coordinate (D0) at ($(B0)+(0,-\H-\Gy)$);
\coordinate (D1) at ($(D0)+(\W,\H)$);
\draw (D0) rectangle (D1);

\coordinate (D_topL) at ($(D0)+(3.2,2.55)$);
\coordinate (D_topR) at ($(D0)+(5.0,2.55)$);
\openseg{D_topL}{D_topR}{K_2^{+}}{above=6pt}

\coordinate (D_botL) at ($(D0)+(3.2,1.10)$);
\coordinate (D_botR) at ($(D0)+(5.0,1.10)$);
\openseg{D_botL}{D_botR}{K_3^{-}}{below=3pt}

% =====================================================
% Thick curve: K1+ (top-left) to K1+ (bottom-left)
% =====================================================

\coordinate (A_S1_mid) at ($(A_L)!0.5!(A_R)$);
\coordinate (C_S1_mid) at ($(C_L1)!0.5!(C_R1)$);

\coordinate (A_S1_mid) at ($(A_L)!0.5!(A_R)$);
\coordinate (C_S1_mid) at ($(C_L1)!0.5!(C_R1)$);

% --- revised K1+ curve: hook around the segments

\coordinate (S1topStart) at ($(A_R)+(-1,0)$);
\coordinate (S1botEnd)   at ($(C_L1)+(1,0)$);

\draw[myarrow]
  (S1topStart)
  .. controls
     ($(A0)+(0,6)$)
     and
     ($(C0)+(-0.55,-2)$)
  .. (S1botEnd);

% =====================================================
% Thick curve K2- (top-right) to K2- (bottom-left)
% =====================================================

\coordinate (S2m_top) at ($(B_botL)!0.5!(B_botR)$);
\coordinate (S2m_bot) at ($(C_L2)!0.5!(C_R2)$);

\def\perpA{1.1}
\def\perpB{1.1}

\draw[myarrow]
  (S2m_top)
  .. controls
     ($(S2m_top)+(0,-\perpA)$)
     and
     ($(S2m_bot)+(0,\perpB)$)
  .. (S2m_bot);

% =====================================================
% Thick curve: connect K2+ (top-right) -> K2+ (bottom-right)
% =====================================================

\coordinate (S2topStart) at ($(B_topR)+(-1,0)$);
\coordinate (S2botEnd)   at ($(D_topL)+(1,0)$);

\draw[myarrow]
  (S2topStart)
  .. controls
     ($(B1)+(0.85,3)$)
     and
     ($(D1)+(0.85,-5)$)
  .. (S2botEnd);

%==============
% Ellipses
%==============

\node[rotate=90] (vell) at ($(D1)+(1,0)$) {\Huge\textbf{$\cdots$}};
\node at ($(D1)+(1,2.5)$) {\Huge$\cdots$};
\node at ($(D0)+(\W+1,0.5)$) {\Huge$\cdots$};

% ---- curve from bottom of vertical ellipsis to top of S3- segment
\coordinate (vellBottom) at ($(vell.west)+(0,-0.05)$);
\coordinate (S3topIn) at ($($(D_botL)!0.5!(D_botR)$)+(0,0)$);

\draw[myarrow]
  (vellBottom)
  .. controls
     ($(vellBottom)+(0,-2.4)$)
     and
     ($(S3topIn)+(0,1)$)
  .. (S3topIn);

%=============
% Box Labels
%=============

\node[anchor=north west] at ($(A0)+(0,\H)$) {$M(1,\alpha)$};

\node[anchor=north west] at ($(B0)+(0,\H)$) {$M(2,\alpha)$};

\node[anchor=north west] at ($(C0)+(0,0.6)$) {$M(1,\beta)$};

\node[anchor=north west] at ($(D0)+(0,0.6)$) {$M(2,\beta)$};

\pgfresetboundingbox
\path[use as bounding box] (C0) rectangle ($(B1)+(1.5,0.2)$);

%\draw[red] (current bounding box.south west) rectangle (current bounding box.north east);

\end{tikzpicture}
%\end{document}

%% file: Theorem_3_Diagram.tex
\begin{tikzpicture}[line cap=round,line join=round]

% ---------- Styling ----------
\tikzset{
  box/.style={draw, very thick, rounded corners=2pt},
  innerbar/.style={draw, very thick},
  thinbar/.style={draw, thick},
  arrow/.style={-Latex, very thick},
  smallcirc/.style={draw, thick},
}

% ---------- Left region: M(1,beta) ----------
% Outer rectangle
\draw[box] (0,0) rectangle (9,5);

% Label
\node[anchor=west] at (0,4.65) {$M(1,\beta)$};

% Three circular components with short "interval" inside
% Bottom-left circle
\draw[very thick] (1.7,1.3) circle (1.0);
\draw[thick] (1.2,1.4) -- (2.2,1.4);
\draw[smallcirc, fill=white] (1.2,1.4) circle (0.08);
\draw[smallcirc, fill=white] (2.2,1.4) circle (0.08);
\node at (1.7,1.75) {$K_2^{-}$};

% Upper circle
\draw[very thick] (2.6,3.3) circle (1.0);
\draw[thick] (2.1,3.3) -- (3.1,3.3);
\draw[smallcirc, fill=white] (2.1,3.3) circle (0.08);
\draw[smallcirc, fill=white] (3.1,3.3) circle (0.08);
\node at (2.6,3.65) {$K_1^{+}$};

% Region O
\draw[very thick] (6.2,2.5) ellipse (1 and 0.9);
\draw[very thick] (6.3,2.3) circle (1.5);

% N in the vertical strip
\draw[thick] (5.75,2.5) -- (6.75,2.5);
\draw[smallcirc, fill=white] (5.75,2.5) circle (0.08);
\draw[smallcirc, fill=white] (6.75,2.5) circle (0.08);
\node at (6.25,2.8) {$S'$};
\node at (6.55,2) {$R$};
\node at (7.35,1.55) {$O$};

% ---------- Connector arrow (down) between the two boxes ----------
\draw[arrow] (10,1.5) -- (10,1);
\draw[arrow] (10,2.5) -- (10,2);
\draw[arrow] (10,3.5) -- (10,3);

% ---------- Right region: R ----------
\draw[box] (10.0,0) rectangle (14.5,5);
\node[anchor=west] at (10,4.65) {M{\"o}bius};

%(5.75,2.5) -- (6.75,2.5)
\draw[thick] (11,2.0) -- (12,2.0);
\draw[smallcirc, fill=white] (11,2.0) circle (0.08);
\draw[smallcirc, fill=white] (12,2.0) circle (0.08);
\node at (11.6,2.3) {$S$};

% Up arrow on right boundary of R
\draw[arrow] (14.5,1) -- (14.5,1.5);
\draw[arrow] (14.5,2) -- (14.5,2.5);
\draw[arrow] (14.5,3) -- (14.5,3.5);

\end{tikzpicture}

%% file: Theorem_4_Diagram.tex
\begin{tikzpicture}[line cap=round,line join=round]

% ---------- Styling ----------
\tikzset{
  box/.style={draw, very thick, rounded corners=2pt},
  innerbar/.style={draw, very thick},
  thinbar/.style={draw, thick},
  arrow/.style={-Latex, very thick},
  smallcirc/.style={draw, thick},
}

% ---------- Left region: M(1,beta) ----------
% Outer rectangle
\draw[box] (0,0) rectangle (9,5);

% Label
\node[anchor=west] at (0,4.65) {$M(1,\beta)$};

% Three circular components with short "interval" inside
% Bottom-left circle
\draw[very thick] (1.7,1.3) circle (1.0);
\draw[thick] (1.2,1.4) -- (2.2,1.4);
\draw[smallcirc, fill=white] (1.2,1.4) circle (0.08);
\draw[smallcirc, fill=white] (2.2,1.4) circle (0.08);
\node at (1.7,1.75) {$K_2^{-}$};

% Upper circle
\draw[very thick] (2.6,3.3) circle (1.0);
\draw[thick] (2.1,3.3) -- (3.1,3.3);
\draw[smallcirc, fill=white] (2.1,3.3) circle (0.08);
\draw[smallcirc, fill=white] (3.1,3.3) circle (0.08);
\node at (2.6,3.65) {$K_1^{+}$};

% Region O
\draw[very thick] (6.2,2.5) ellipse (1 and 0.9);
\draw[very thick] (6.3,2.3) circle (1.5);

% N in the vertical strip
\draw[thick] (5.75,2.5) -- (6.75,2.5);
\draw[smallcirc, fill=white] (5.75,2.5) circle (0.08);
\draw[smallcirc, fill=white] (6.75,2.5) circle (0.08);
\node at (6.25,2.8) {$S'$};
\node at (6.55,2) {$R$};
\node at (7.35,1.55) {$O$};

% ---------- Right region: R ----------
\draw[box] (10.0,0) rectangle (14.5,5);
\node[anchor=west] at (10,4.65) {Minkowski};

%(5.75,2.5) -- (6.75,2.5)
\draw[thick] (11,2.0) -- (12,2.0);
\draw[smallcirc, fill=white] (11,2.0) circle (0.08);
\draw[smallcirc, fill=white] (12,2.0) circle (0.08);
\node at (11.6,2.3) {$S$};

\end{tikzpicture}

%% file: BD_ClotheslineOS.tex
%\documentclass[tikz,border=4pt]{standalone}
%\usepackage{tikz}
%\usetikzlibrary{arrows.meta}
%\usetikzlibrary{calc}

%\begin{document}
\begin{tikzpicture}[
  line cap=round,
  line join=round,
  line width=0.9pt
]

% ---------- panel geometry ----------
\def\W{6.2}
\def\H{3.9}
\def\Gx{0.7}
\def\Gy{0.8}

% ---------- helpers ----------
\newcommand{\openseg}[4]{%
  \draw (#1) -- (#2);
  \draw[fill=white] (#1) circle (0.08);
  \draw[fill=white] (#2) circle (0.08);
  \node[#4] at ($(#1)!0.5!(#2)$) {$#3$};
}

\newcommand{\lightcone}[3]{%
  \begin{scope}
    \clip (#2) rectangle (#3);
    \draw[dash pattern=on 2.2pt off 2.2pt]
      (#1) -- ++(225:20);
    \draw[dash pattern=on 2.2pt off 2.2pt]
      (#1) -- ++(315:20);
  \end{scope}
}

\tikzset{
  myarrow/.style={
    {Stealth[length=4mm,width=4mm]}-{Stealth[length=4mm,width=4mm]},
    line width=0.9pt
  }
}

% =========================
% Box 1 (top-left)
% =========================
\coordinate (A0) at (0,0);
\coordinate (A1) at ($(A0)+(\W,\H)$);
\draw (A0) rectangle (A1);

% K_1^+ -- shortened to match the scaled K_1^+
% in the replacement diagram
\coordinate (A_L) at ($(A0)+(1.0,2.35)$);
\coordinate (A_R) at ($(A0)+(1.93,2.35)$);
\openseg{A_L}{A_R}{K_1^{+}}{below=3pt}

\coordinate (p1) at ($(A0)+(4.7,3.1)$);
\fill (p1) circle (0.08);
%\node at ($(p1)+(0.5,0.12)$) {$p_1$};

\node at ($(p1)+(0,-1.5)$) {$C_1$};

\lightcone{p1}{A0}{A1}

% =========================
% Box 2 (top-right)
% =========================
\coordinate (B0) at ($(A0)+(\W+\Gx,0)$);
\coordinate (B1) at ($(B0)+(\W,\H)$);
\draw (B0) rectangle (B1);

\coordinate (B_topL) at ($(B0)+(3.0,2.75)$);
\coordinate (B_topR) at ($(B0)+(4.8,2.75)$);
\openseg{B_topL}{B_topR}{K_2^{+}}{above=6pt, xshift=-10pt}

\coordinate (B_botL) at ($(B0)+(3.0,2.05)$);
\coordinate (B_botR) at ($(B0)+(4.8,2.05)$);
\openseg{B_botL}{B_botR}{K_2^{-}}{above=0pt}

\coordinate (p2) at ($(B0)+(1.9,2.55)$);
\fill (p2) circle (0.08);
%\node at ($(p2)+(-0.5,0.25)$) {$p_2$};

\node at ($(p2)+(0,-1.3)$) {$C_2$};

\lightcone{p2}{B0}{B1}

% =========================
% Box 3 (bottom-left)
% =========================
\coordinate (C0) at ($(A0)+(0,-\H-\Gy)$);
\coordinate (C1) at ($(C0)+(\W,\H)$);
%\draw (C0) rectangle (C1);

\coordinate (C_L1) at ($(C0)+(1.2,1.95)$);
\coordinate (C_R1) at ($(C0)+(3.2,1.95)$);
%\openseg{C_L1}{C_R1}{K_1^{+}}{above=3pt}

\coordinate (C_L2) at ($(C0)+(4.0,1.95)$);
\coordinate (C_R2) at ($(C0)+(5.8,1.95)$);
%\openseg{C_L2}{C_R2}{K_2^{-}}{below=3pt}

% =========================
% Box 4 (bottom-right)
% =========================
\coordinate (D0) at ($(B0)+(0,-\H-\Gy)$);
\coordinate (D1) at ($(D0)+(\W,\H)$);
\draw (D0) rectangle (D1);

\coordinate (D_topL) at ($(D0)+(3.2,2.55)$);
\coordinate (D_topR) at ($(D0)+(5.0,2.55)$);
\openseg{D_topL}{D_topR}{K_2^{+}}{above=6pt}

\coordinate (D_botL) at ($(D0)+(3.2,1.10)$);
\coordinate (D_botR) at ($(D0)+(5.0,1.10)$);
\openseg{D_botL}{D_botR}{K_3^{-}}{below=3pt}

% =====================================================
% Thick curve: K1+ (top-left) to K1+ (bottom-left)
% =====================================================

% Start at the midpoint of the top K_1^+ segment
\coordinate (S1topStart) at ($(A_L)!0.5!(A_R)$);

% Midpoint of the K_1^+ segment in the scaled replacement diagram.
% The unscaled segment runs from (2.2,4.4) to (3.4,4.4),
% so its midpoint is (2.8,4.4).
\coordinate (S1botEnd) at
  ($(C0)+(2.8*0.775,4.4*0.78)$);

\draw[myarrow]
  (S1topStart)
  .. controls
     ($(A0)+(0,6)$)
     and
     ($(C0)+(-0.55,-2)$)
  .. (S1botEnd);

% =====================================================
% Thick curve K2- (top-right) to K2- (bottom-left)
% =====================================================

% Midpoint of the top K_2^- segment
\coordinate (S2m_top) at ($(B_botL)!0.5!(B_botR)$);

% Midpoint of the K_2^- segment in the scaled replacement diagram.
% The unscaled segment runs from (4.6,4.4) to (5.8,4.4),
% so its midpoint is (5.2,4.4).
\coordinate (S2m_bot) at
  ($(C0)+(5.2*0.775,4.4*0.78)$);

\def\perpA{1.1}
\def\perpB{1.1}

\draw[myarrow]
  (S2m_top)
  .. controls
     ($(S2m_top)+(0,-\perpA)$)
     and
     ($(S2m_bot)+(0,\perpB)$)
  .. (S2m_bot);

% =====================================================
% Thick curve: K2+ (top-right) -> K2+ (bottom-right)
% =====================================================

\coordinate (S2topStart) at ($(B_topR)+(-1,0)$);
\coordinate (S2botEnd)   at ($(D_topL)+(1,0)$);

\draw[myarrow]
  (S2topStart)
  .. controls
     ($(B1)+(0.85,3)$)
     and
     ($(D1)+(0.85,-5)$)
  .. (S2botEnd);

%==============
% Ellipses
%==============

\node[rotate=90] (vell) at ($(D1)+(1,0)$) {\Huge\textbf{$\cdots$}};
\node at ($(D1)+(1,2.5)$) {\Huge$\cdots$};
\node at ($(D0)+(\W+1,0.5)$) {\Huge$\cdots$};

% ---- curve from bottom of vertical ellipsis to top of S3- segment
\coordinate (vellBottom) at ($(vell.west)+(0,-0.05)$);
\coordinate (S3topIn) at ($($(D_botL)!0.5!(D_botR)$)+(0,0)$);

\draw[myarrow]
  (vellBottom)
  .. controls
     ($(vellBottom)+(0,-2.4)$)
     and
     ($(S3topIn)+(0,1)$)
  .. (S3topIn);

%=============
% Box Labels
%=============

\node[anchor=north west] at ($(A0)+(0,\H)$) {$M(1,\alpha)$};

\node[anchor=north west] at ($(B0)+(0,\H)$) {$M(2,\alpha)$};

\node[anchor=north west] at ($(D0)+(0,0.6)$) {$M(2,\beta)$};

% Bounding box
\pgfresetboundingbox
\path[use as bounding box] (C0) rectangle ($(B1)+(1.5,0.2)$);

% =====================================================
% Replacement for Box 3
% =====================================================

\begin{scope}[
    shift={(C0)},
    xscale=0.775,
    yscale=0.78
]

    \def\BoxTop{5}
    \def\BoxWidth{8}
    \def\BoxCornerX{0}
    \def\BoxCornerY{0}
    \def\TopApexY{0.5*\BoxTop}
    \def\TopApexX{4}
    \def\BotApexY{0}
    \def\BotApexX{4}
    \def\labeloffsettop{0.4}

    % Keep circles at their original physical size.
    % The node positions are still transformed by the
    % surrounding xscale/yscale.
    \newcommand{\fixedcircle}[1]{%
      \node[
        circle,
        draw,
        fill=white,
        minimum size=0.16cm,
        inner sep=0pt
      ] at #1 {};
    }

    % --------- Making replacement for Box 3 ---------

    \begin{scope}[xshift=0cm]

      % bounding box
      \draw
        (\BoxCornerX,\BoxCornerY)
        rectangle
        (\BoxWidth,\BoxTop);

      % wedges
      \coordinate (TopApex) at (\TopApexX,\TopApexY);
      \coordinate (BotApex) at (\BotApexX,\BotApexY);

      % top wedge
      \draw
        (0.9,\BoxTop) -- (TopApex) -- (7.1,\BoxTop);

      % singularity
      \draw[dashed]
        (TopApex) -- (BotApex);

      \node[rotate=90]
        at (\TopApexX + \labeloffsettop,0.5*\TopApexY)
        {Singularity};

      % labels
      \node
        at (\TopApexX,\TopApexY+1)
        {FLRW};

      \node[anchor=west]
        at (0.2,\labeloffsettop)
        {Schwarzschild};

      % K_1^+ slit
      \draw
        (\TopApexX-1.8,\BoxTop-0.6)
        -- (3.4,\BoxTop-0.6);

      % K_2^- slit
      \draw
        (\TopApexX+0.6,\BoxTop-0.6)
        -- (5.8,\BoxTop-0.6);

      % Open endpoints for K_1^+
      \fixedcircle{(\TopApexX-1.8,\BoxTop-0.6)}
      \fixedcircle{(3.4,\BoxTop-0.6)}

      % Open endpoints for K_2^-
      \fixedcircle{(\TopApexX+0.6,\BoxTop-0.6)}
      \fixedcircle{(5.8,\BoxTop-0.6)}

      % Labels
      \node
        at (2.8,\BoxTop - 0.7 + \labeloffsettop)
        {$K_1^+$};

      \node
        at (5.2,\BoxTop - 0.6 - \labeloffsettop)
        {$K_2^-$};

    % gamma

    \def\gammax{7.3}
    
    \draw
      (\gammax,\BoxTop) -- (\gammax,0);

    \node[anchor=west]
      at (\gammax+0.1,0.75*\BoxTop)
      {$\gamma$};

    \end{scope}

\end{scope}

\end{tikzpicture}
%\end{document}

%% file: BD_MinkowskiSch.tex
\begin{tikzpicture}[
  line cap=round,
  line join=round,
  line width=0.9pt
]

% ---------- panel geometry ----------
\def\W{6.2}
\def\H{3.9}
\def\Gx{0.7}
\def\Gy{0.8}

% ---------- helpers ----------
\newcommand{\openseg}[4]{%
  \draw (#1) -- (#2);
  \node[
    circle,
    draw,
    fill=white,
    minimum size=0.16cm,
    inner sep=0pt
  ] at (#1) {};
  \node[
    circle,
    draw,
    fill=white,
    minimum size=0.16cm,
    inner sep=0pt
  ] at (#2) {};
  \node[#4] at ($(#1)!0.5!(#2)$) {$#3$};
}

\newcommand{\lightcone}[3]{%
  \begin{scope}
    \clip (#2) rectangle (#3);
    \draw[dash pattern=on 2.2pt off 2.2pt]
      (#1) -- ++(225:20);
    \draw[dash pattern=on 2.2pt off 2.2pt]
      (#1) -- ++(315:20);
  \end{scope}
}

\tikzset{
  myarrow/.style={
    {Stealth[length=4mm,width=4mm]}-{Stealth[length=4mm,width=4mm]},
    line width=0.9pt
  },
  arrow/.style={
    -Latex,
    line width=0.9pt
  }
}

% =====================================================
% Box 1 (top-left)
% =====================================================

\coordinate (A0) at (0,0);
\coordinate (A1) at ($(A0)+(\W,\H)$);
\draw (A0) rectangle (A1);

\coordinate (A_L) at ($(A0)+(1.0,2.35)$);
\coordinate (A_R) at ($(A0)+(1.93,2.35)$);
\openseg{A_L}{A_R}{K_1^{+}}{below=3pt}

\coordinate (p1) at ($(A0)+(4.7,3.1)$);
\fill (p1) circle (0.08);

\node at ($(p1)+(0,-1.5)$) {$C_1$};

\lightcone{p1}{A0}{A1}

% =====================================================
% Box 2 (top-right)
% =====================================================

\coordinate (B0) at ($(A0)+(\W+\Gx,0)$);
\coordinate (B1) at ($(B0)+(\W,\H)$);
\draw (B0) rectangle (B1);

\coordinate (B_topL) at ($(B0)+(3.0,2.75)$);
\coordinate (B_topR) at ($(B0)+(4.8,2.75)$);
\openseg{B_topL}{B_topR}{K_2^{+}}{above=6pt, xshift=-10pt}

\coordinate (B_botL) at ($(B0)+(3.0,2.05)$);
\coordinate (B_botR) at ($(B0)+(8.2-4,2.05)$);
\openseg{B_botL}{B_botR}{K_2^{-}}{above=0pt}

\coordinate (p2) at ($(B0)+(1.9,2.55)$);
\fill (p2) circle (0.08);

\node at ($(p2)+(0,-1.3)$) {$C_2$};

\lightcone{p2}{B0}{B1}

% =====================================================
% Box 3 location
%
% The old Box 3 rectangle is omitted.  It is replaced
% below by the first box from the new three-box diagram.
% =====================================================

\coordinate (C0) at ($(A0)+(0,-\H-\Gy)$);
\coordinate (C1) at ($(C0)+(\W,\H)$);

% =====================================================
% Box 4 (bottom-right of clothesline)
% =====================================================

\coordinate (D0) at ($(B0)+(0,-\H-\Gy)$);
\coordinate (D1) at ($(D0)+(\W,\H)$);
\draw (D0) rectangle (D1);

\coordinate (D_topL) at ($(D0)+(3.2,2.55)$);
\coordinate (D_topR) at ($(D0)+(5.0,2.55)$);
\openseg{D_topL}{D_topR}{K_2^{+}}{above=6pt}

\coordinate (D_botL) at ($(D0)+(3.2,1.10)$);
\coordinate (D_botR) at ($(D0)+(5.0,1.10)$);
\openseg{D_botL}{D_botR}{K_3^{-}}{below=3pt}

% =====================================================
% Thick curve: K1+ (top-left) to replacement K1+
% =====================================================

\coordinate (S1topStart) at ($(A_L)!0.5!(A_R)$);

% Midpoint of K_1^+ in replacement:
% local endpoints (2.2,4.4) and (3.4,4.4)
\coordinate (S1botEnd) at
  ($(C0)+(2.8*0.775,4.4*0.78)$);

\draw[myarrow]
  (S1topStart)
  .. controls
     ($(A0)+(0,6)$)
     and
     ($(C0)+(-0.55,-2)$)
  .. (S1botEnd);

% =====================================================
% Thick curve K2- (top-right) to replacement K2-
% =====================================================

\coordinate (S2m_top) at ($(B_botL)!0.5!(B_botR)$);

% Midpoint of K_2^- in replacement:
% local endpoints (4.6,4.4) and (5.8,4.4)
\coordinate (S2m_bot) at
  ($(C0)+(5.2*0.775,4.4*0.78)$);

\def\perpA{1.1}
\def\perpB{1.1}

\draw[myarrow]
  (S2m_top)
  .. controls
     ($(S2m_top)+(0,-\perpA)$)
     and
     ($(S2m_bot)+(0,\perpB)$)
  .. (S2m_bot);

% =====================================================
% Thick curve K2+ (top-right) -> K2+ (bottom-right)
% =====================================================

\coordinate (S2topStart) at ($(B_topR)+(-1,0)$);
\coordinate (S2botEnd)   at ($(D_topL)+(1,0)$);

\draw[myarrow]
  (S2topStart)
  .. controls
     ($(B1)+(0.85,3)$)
     and
     ($(D1)+(0.85,-5)$)
  .. (S2botEnd);

% =====================================================
% Ellipses
% =====================================================

\node[rotate=90] (vell) at ($(D1)+(1,0)$) {\Huge\textbf{$\cdots$}};
\node at ($(D1)+(1,2.5)$) {\Huge$\cdots$};
\node at ($(D0)+(\W+1,0.5)$) {\Huge$\cdots$};

\coordinate (vellBottom) at ($(vell.west)+(0,-0.05)$);
\coordinate (S3topIn) at ($($(D_botL)!0.5!(D_botR)$)+(0,0)$);

\draw[myarrow]
  (vellBottom)
  .. controls
     ($(vellBottom)+(0,-2.4)$)
     and
     ($(S3topIn)+(0,1)$)
  .. (S3topIn);

% =====================================================
% Box Labels
% =====================================================

\node[anchor=north west] at ($(A0)+(0,\H)$) {$M(1,\alpha)$};

\node[anchor=north west] at ($(B0)+(0,\H)$) {$M(2,\alpha)$};

\node[anchor=north west] at ($(D0)+(0,0.6)$) {$M(2,\beta)$};

% =====================================================
% REPLACEMENT BOX 3
%
% First box from the new three-box diagram.
% =====================================================

\begin{scope}[
    shift={(C0)},
    xscale=0.775,
    yscale=0.78
]

    \def\BoxTop{5}
    \def\BoxWidth{8}
    \def\BoxCornerX{0}
    \def\BoxCornerY{0}
    \def\TopApexY{0.5*\BoxTop}
    \def\TopApexX{4}
    \def\BotApexY{0}
    \def\BotApexX{4}
    \def\labeloffsettop{0.4}

    % Fixed-size endpoint circles
    \newcommand{\fixedcircle}[1]{%
      \node[
        circle,
        draw,
        fill=white,
        minimum size=0.16cm,
        inner sep=0pt
      ] at #1 {};
    }

    % Bounding box
    \draw
      (\BoxCornerX,\BoxCornerY)
      rectangle
      (\BoxWidth,\BoxTop);

    % Wedges
    \coordinate (TopApex) at (\TopApexX,\TopApexY);
    \coordinate (BotApex) at (\BotApexX,\BotApexY);

    \draw
      (0.9,\BoxTop) -- (TopApex) -- (7.1,\BoxTop);

    % Singularity
    \draw[dashed]
      (TopApex) -- (BotApex);

    \node[rotate=90]
      at (\TopApexX + \labeloffsettop,0.5*\TopApexY)
      {Singularity};

    % Labels
    \node
      at (\TopApexX,\TopApexY+1)
      {FLRW};

    \node[anchor=west]
      at (0.2,\labeloffsettop)
      {Schwarzschild};

    % K_1^+ slit
    \draw
      (\TopApexX-1.8,\BoxTop-0.6)
      -- (3.4,\BoxTop-0.6);

    % K_2^- slit
    \draw
      (\TopApexX+0.6,\BoxTop-0.6)
      -- (5.8,\BoxTop-0.6);

    % Endpoint circles
    \fixedcircle{(\TopApexX-1.8,\BoxTop-0.6)}
    \fixedcircle{(3.4,\BoxTop-0.6)}

    \fixedcircle{(\TopApexX+0.6,\BoxTop-0.6)}
    \fixedcircle{(5.8,\BoxTop-0.6)}

    % Labels
    \node
      at (2.8,\BoxTop - 0.7 + \labeloffsettop)
      {$K_1^+$};

    \node
      at (5.2,\BoxTop - 0.6 - \labeloffsettop)
      {$K_2^-$};

    % S_1
    \draw
      (5.5,0.65*\TopApexY)
      -- (6.5,0.65*\TopApexY);

    \fixedcircle{(5.5,0.65*\TopApexY)}
    \fixedcircle{(6.5,0.65*\TopApexY)}

    \node
      at (6,0.65*\TopApexY+\labeloffsettop)
      {$S_1$};

\end{scope}

% =====================================================
% LOWER ROW
%
% The following two regions are placed independently
% below the clothesline.
% =====================================================

% -----------------------------------------------------
% Lower-left: M(1,\beta)
% -----------------------------------------------------

\coordinate (E0) at ($(C0)+(0,-\H-\Gy)$);

\begin{scope}[
    shift={(E0)},
    xscale=0.775,
    yscale=0.78
]

    \def\BoxTopB{5}
    \def\BoxWidthB{8}
    \def\TopApexXB{4}
    \def\BotApexYB{2.0}
    \def\labeloffsetB{0.4}

    \newcommand{\fixedcircleB}[1]{%
      \node[
        circle,
        draw,
        fill=white,
        minimum size=0.16cm,
        inner sep=0pt
      ] at #1 {};
    }

    % Bounding box
    \draw
      (0,0) rectangle (\BoxWidthB,\BoxTopB);

    % Wedge
    %\coordinate (TopApexB) at (\TopApexXB,\BoxTopB);
    %\coordinate (BotApexB) at (\TopApexXB,\BotApexYB);

    %\draw
    %  (0.9,0) -- (BotApexB) -- (7.1,0);

    % Wedge
    \coordinate (TopApexB) at (\TopApexXB,\BoxTopB);
    \coordinate (BotApexB) at (\TopApexXB,2.5);

    \draw
      (2.8,0) -- (BotApexB) -- (5.2,0);

    % Singularity
    \draw[dashed]
      (TopApexB) -- (BotApexB);

    \node[rotate=90]
      at (\TopApexXB+\labeloffsetB,1.3*\BoxTopB/2+0.5)
      {Singularity};

    % Labels
    \node[anchor=west]
      at (0.2,\BoxTopB-\labeloffsetB)
      {Schwarzschild};

    \node
      at (4,0.5)
      {Minkowski};

    % S_2
    %\draw
    %  (3.5,1.0) -- (4.5,1.0);

    %\fixedcircleB{(3.5,1.0)}
    %\fixedcircleB{(4.5,1.0)}

    %\openseg{3.5}{1.0}{S_2}{below=3pt}

    \coordinate (S2_L) at (3.5,1.0);
    \coordinate (S2_R) at (4.5,1.0);
    
    \openseg{S2_L}{S2_R}{S_2}{above=1pt}

    %\node
    %  at (4,1.3)
    %  {$S_2$};

    % S_1
    \draw
      (5.5,0.85*\BoxTopB/2)
      -- (6.5,0.85*\BoxTopB/2);

    \fixedcircleB{(5.5,0.85*\BoxTopB/2)}
    \fixedcircleB{(6.5,0.85*\BoxTopB/2)}

    \node
      at (6,0.85*\BoxTopB/2+\labeloffsetB)
      {$S_1$};

    % gamma

    %\def\gammax{7}
    
    %\draw
    %  (\gammax,\BoxTopB) -- (\gammax,0);

    %\node[anchor=west]
    %  at (\gammax+\labeloffsetB-0.4,0.75*\BoxTopB)
    %  {$\gamma$};

\end{scope}

% -----------------------------------------------------
% Lower-right: R
% -----------------------------------------------------

% The R box is placed to the right of the M(1,beta)
% box in the lower row.

\coordinate (F0) at ($(E0)+(7.0,0)$);

\begin{scope}[
    shift={(F0)},
    xscale=0.775,
    yscale=0.78
]

    \def\RWidth{8}
    \def\RHeight{5}
    \def\labeloffsetR{0.4}

    \newcommand{\fixedcircleR}[1]{%
      \node[
        circle,
        draw,
        fill=white,
        minimum size=0.16cm,
        inner sep=0pt
      ] at #1 {};
    }

    % R region
    \draw
      (0,0) rectangle (\RWidth,\RHeight);

    \node[anchor=west]
      at (0,\RHeight-0.35)
      {M{\"o}bius};

    % S_2
    \draw
      (1,2.0) -- (2,2.0);

    \fixedcircleR{(1,2.0)}
    \fixedcircleR{(2,2.0)}

    \node
      at (1.5,2.0+\labeloffsetR)
      {$S_2$};

    % Up arrows on right boundary of R
    \draw[arrow]
      (\RWidth,1) -- (\RWidth,1.5);

    \draw[arrow]
      (\RWidth,2) -- (\RWidth,2.5);

    \draw[arrow]
      (\RWidth,3) -- (\RWidth,3.5);

    % Down arrows on left boundary of R
    \draw[arrow]
      (0,1.5) -- (0,1);

    \draw[arrow]
      (0,2.5) -- (0,2);

    \draw[arrow]
      (0,3.5) -- (0,3);

\end{scope}

% =====================================================
% Bounding box
% =====================================================

%\pgfresetboundingbox
%\path[
%  use as bounding box
%]
%  (C0)
%  rectangle
%  ($(F0)+(0.775*4.5,0.78*5)$);

\end{tikzpicture}

%% file: references.bib
@ARTICLE{Israel:1966,
       author = {Werner Israel},
        title = "{Singular hypersurfaces and thin shells in general relativity}",
      journal = {Nuovo Cimento B},
      volume = {66},
         year = "1966",
         number = {1},
         pages =
         
         {1-14}
}

@ARTICLE{CollinsHawking:1973,
       author = {C.B. Collins and Stephen Hawking},
        title = "{Why is the Universe isotropic?}",
      journal = {The Astrophysical Journal},
         year = "1973",
         volume = {180},
         pages = {317-334}
}

@book{Belot:2023,
    author    = {Gordon Belot},
    title     = "{Accelerating Expansion: Philosophy and Physics with a Positive Cosmological Constant}",
    year      = "2023",
    publisher = "Oxford University Press"
}

@book{HawkingEllis:1973,
    author    = {Stephen Hawking and George Ellis},
    title     = "{The Large-Scale Structure of Space-Time}",
    year      = "1973",
    publisher = "Cambridge University Press"
}

@book{Raychaudhuri:1979,
    author    = {Amal Raychaudhuri},
    title     = "{Theoretical Cosmology}",
    year      = "1979",
    publisher = "Clarendon Press"
}

@book{EllisMaartensMacCallum:2012,
    author    = {George Ellis and Roy Maartens and Malcolm MacCallum},
    title     = "{Relativistic Cosmology}",
    year      = "2012",
    publisher = "Cambridge University Press"
}

@book{Poisson:2004,
    author    = {Erik Poisson},
    title     = "A Relativist's Toolkit: The Mathematics of Black Hole Mechanics",
    year      = "2004",
    publisher = "Cambridge University Press"
}

@article{Senovilla:1996, 
	title="{Towards realistic singularity-free cosmological models}", 
	volume = {53},
    issue = {4},
	journal = {Physical Review D}, 
	author={J.M.M. Senovilla}, 
	year={1996}, 
	pages={1799-1807}
}

@article{Senovilla:1990, 
	title="{New Class of Inhomogeneous Cosmological Perfect-Fluid Solutions without Big-Bang Singularity}", 
	volume = {64},
    issue = {19},
	journal = {Physical Review Letters}, 
	author={J.M.M. Senovilla}, 
	year={1990}, 
	pages={2219-2221}
}

@incollection{Senovilla:1993,
  author = {J.M.M. Senovilla},
  title = {Singularity-free spacetimes}, 
  editor = {F. J. Chinea and L. M. González-Romero}, 
  booktitle = "{Rotating Objects and Relativistic Physics}", 
  publisher = {Springer},
  year = "1993",
  pages = {185-193}
}

@incollection{Geroch:1971,
  author = {Robert Geroch},
  title = {Space-Time Structure From a Global Viewpoint}, 
  editor = {B.K. Sachs}, 
  booktitle = "{General Relativity and Cosmology}",
  publisher = {Academic Press},
  year = "1971",
  pages = {71-103}
}

@incollection{GerochHorowitz:1979,
  author = {Robert Geroch and Gary Horowitz},
  title = {Global structure of spacetime}, 
  editor = {Stephen Hawking and Werner Israel}, 
  booktitle = "{General Relativity: An Einstein Centenary Survey}",
  publisher = {Cambridge University Press},
  year = "1979",
  pages = {212-293}
}

@incollection{Earman:1999,
  author = {John Earman},
  title = "{The Hawking–Penrose singularity theorems: History and implications}", 
  editor = {H. Goenner and J. Renn and J. Ritter and T. Sauer}, 
  booktitle = "{The Expanding Worlds of General Relativity}", 
  publisher = {Birkhäuser},
  year = "1999",
  pages = {235–267}
}

@article{ClarkeDray:1987,
    title = "{Junction conditions for null hypersurfaces}",
    journal = {Classical and Quantum Gravity},
    volume = {4},
    pages = {265-275},
    year = {1987},
    author = {C.J.S. Clarke and Tevian Dray}
}

@book{MisnerThorneWheeler:1973,
    author    = "Charles Misner and Kip Thorne and John Wheeler",
    title     = "Gravitation",
    year      = "1973",
    publisher = "W.H. Freeman and Company"
}

@article{Hawking:1969,
     author = {Stephen Hawking},
     journal = {Proceedings of the Royal Society of London. A. Mathematical and Physical Sciences},
     volume = {308},
     number = {1494},
     pages = {433-435},
     title = "{The Existence of Cosmic Time Functions}",
     year = {1969}
}

@article{Hawking:1965,
     author = {Stephen Hawking},
     journal = {Physical Review Letters},
     volume = {15},
     number = {17},
     pages = {689-690},
     title = "{Occurrence of Singularities in Open Universes}",
     year = {1965}
}

@article{Hawking:1966,
     author = {Stephen Hawking},
     journal = {Proceedings of the Royal Society of London, Series A},
     volume = {294},
     number = {1439},
     pages = {511-521},
     title = "{The Occurrence of Singularities in Cosmology}",
     year = {1966}
}

@article{Hawking_Penrose:1970,
     author = {Stephen Hawking and Roger Penrose},
     journal = {Proceedings of the Royal Society of London, Series A},
     volume = {314},
     pages = {529-548},
     title = "{The singularities of gravitational collapse and cosmology}",
     year = {1970}
}

@article{Geroch:1966,
     author = {Robert Geroch},
     journal = {Physical Review Letters},
     volume = {17},
     number = {8},
     pages = {445-447},
     title = "{Singularities in Closed Universes}",
     year = {1966}
}

@article{Senovilla:2007,
     author = {J.M.M. Senovilla},
     journal = {Pramana},
     volume = {69},
     pages = {31-47},
     title = "{A singularity theorem based on spatial averages}",
     year = {2007}
}

@article{DrobovTegai:2013,
  title = {Dust thin shell limit of a thick wall},
  author = {Drobov, I. V. and Tegai, S. Ph.},
  journal = {Physical Review D},
  volume = {87},
  issue = {2},
  pages = {024025},
  numpages = {4},
  year = {2013}
}

@article{KhakshourniaMansouri:2002,
  title = {Dynamics of General Relativistic Spherically Symmetric Dust Thick Shells},
  author = {S. Khakshournia and R. Mansouri},
  journal = {General Relativity and Gravitation},
  volume = {34},
  issue = {11},
  pages = {1847-1853},
  year = {2002}
}

@article{KhosraviKhakshourniaMansouri:2006,
  title = {Evolution of thick shells in curved spacetimes},
  author = {Sh. Khosravi and S. Khakshournia and R. Mansouri},
  journal = {Classical and Quantum Gravity},
  volume = {23},
  issue = {20},
  pages = {5927-5939},
  year = {2006}
}

@book{Wald:1984,
    author    = {Robert Wald},
    title     = "{General Relativity}",
    year      = "1984",
    publisher = "The University of Chicago Press"
}

@book{Goodman:1955,
    author    = {Nelson Goodman},
    title     = "{Fact, Fiction, \& Forecast}",
    year      = "1995",
    publisher = "Harvard University Press"
}

@ARTICLE{Fayos:1991,
   author = {F. Fayos and X. Jaen and E. Llanta and J.M.M. Senovilla},
    title = "{Matching of the Vaidya and Robertson-Walker metric}",
  journal = {Classical and Quantum Gravity},
     year = 1991,
   volume = 8,
  pages = {2057--2068}
}

@ARTICLE{FewsterGalloway:2011,
   title = {Singularity theorems from weakened energy conditions},
    author = {Christopher J Fewster and Gregory J Galloway},
  journal = {Classical and Quantum Gravity},
     year = 2011,
   volume = 28,
    issue = 12,
  pages = {1-18}
}

@ARTICLE{FewsterKontou:2020,
   title = {A new derivation of singularity theorems with weakened energy hypotheses},
    author = {Fewster, Christopher and Kontou, Eleni-Alexandra},
  journal = {Classical and Quantum Gravity},
     year = 2020,
   volume = 37,
    issue = 6,
  pages = {1-31}
}

@ARTICLE{Tipler:1978,
   title = {Energy conditions and spacetime singularities},
    author = {Tipler, Frank J.},
  journal = {Physical Review D},
     year = 1978,
   volume = 17,
    issue = 10,
  pages = {2521--2528}
}

@ARTICLE{Brown_etal:2018,
   title = "{A singularity theorem for Einstein–Klein–Gordon theory}",
    author = {Brown, Peter and Fewster, Christopher and Kontou, Eleni-Alexandra},
  journal = {General Relativity and Gravitation},
     year = 2018,
   volume = 50,
    issue = 121,
  pages = {1-24}
}

@ARTICLE{Wall:2013,
   title = {The generalized second law implies a quantum singularity theorem},
    author = {Wall, Aaron},
  journal = {Classical and Quantum Gravity},
     year = 2013,
   volume = 30,
  pages = {1-35}
}

@ARTICLE{Senovilla:2021,
   title = {A critical appraisal of the singularity theorems},
    author = {J.M.M. Senovilla},
  journal = {Philosophical Transactions A},
     year = 2021,
   volume = 380,
  pages = {1-16}
}

@ARTICLE{Fewster:2022,
   author = {Christopher J Fewster and Eleni-Alexandra Kontou},
    title = "{A semiclassical singularity theorem}",
  journal = {Classical and Quantum Gravity},
     year = 2022,
   volume = 39,
  pages = {1-32}
}

@ARTICLE{Manchak:2009,
   author = {Manchak, J.B. },
    title = "{Can we know the global structure of spacetime?}",
  journal = {Studies in History and Philosophy of Modern Physics},
     year = 2009,
   volume = 40,
  pages = {53--56}
}

@ARTICLE{Manchak:2011,
   author = {J.B. Manchak},
    title = "{What Is a Physically Reasonable Space-Time?}",
  journal = {Philosophy of Science},
     year = 2011,
   volume = 78,
   number = 3,
  pages = {410--420}
}

@article{Manchak:2016b,
     author = {J.B. Manchak},
     journal = {Philosophy of Science},
     volume = {83},
     pages = {1050-1058},
     title = "{On G{\"o}del and the Ideality of Time}",
     year = {2016}
}

@ARTICLE{Glymour:1972,
       author = {Clark Glymour},
        title = "{Topology, Cosmology and Convention}",
      journal = "{Synthese}",
         year = "1972",
       volume = {24},
       issue = {1/2},
       pages = {195-218}
}

@ARTICLE{Ellis:1971,
       author = {George Ellis},
        title = "{Topology and Cosmology}",
      journal = "{General Relativity and Gravitation}",
         year = "1971",
       volume = {2},
       issue = {7},
       pages = {7-21}
}

@ARTICLE{Marder:1964,
       author = {L. Marder},
        title = "{On space-times with bounded empty regions}",
      journal = "{Mathematical Proceedings of the Cambridge Philosophical Society}",
         year = "1964",
       volume = {60},
       issue = {1},
       pages = {97 - 103}
}

@book{Manchak:2020, 
    title = "{Global Spacetime Structure}", 
    publisher={Cambridge University Press}, 
    author={J.B. Manchak}, 
    year={2020}
}

@book{Art:2023, 
    title = "{The Art of Gluing Space-Time Manifolds: Methods and Applications}", 
    publisher={Springer}, 
    author={Samad Khakshournia and Reza Mansouri}, 
    year={2023}
}

@book{Earman:1995, 
    title = "Bangs, Crunches, Whimpers, and Shrieks", 
    publisher={Oxford University Press}, 
    author={John Earman}, 
    year={1995},
    address = {Oxford}
}

@InCollection{sep-spacetime-singularities,
	author       =	{Curiel, Erik},
	title        =	"Singularities and Black Holes",
	booktitle    =	"The {Stanford} Encyclopedia of Philosophy",
	editor       =	{Edward N. Zalta},
	howpublished =	{\url{https://plato.stanford.edu/archives/spr2021/entries/spacetime-singularities/}},
	year         =	{2021},
	edition      =	{{S}pring 2021},
	publisher    =	{Metaphysics Research Lab, Stanford University},
    address = {n.p.}
}

@article{COMPACT:2026uzi,
    author = "Copi, Craig J. and others",
    collaboration = "COMPACT",
    title = "{The Topology of the Universe}",
    eprint = "2606.24886",
    archivePrefix = "arXiv",
    primaryClass = "astro-ph.CO",
    month = "6",
    year = "2026"
}

@InCollection{Glymour:1977,
	author       =	{Clark Glymour},
	title        =	"{Indistinguishable Space-Times and the Fundamental Group}",
	booktitle    =	"{Foundations of Space-Time Theories}",
    Editors = {John Earman and John J. Stachel},
    publisher = "University of Minnesota Press",
    year = "1977",
    pages = {50-60}
}

@InCollection{Malament:1977,
	author       =	{David Malament},
	title        =	"{Observationally Indistinguishable Spacetimes}",
	booktitle    =	"{Foundations of Space-Time Theories}",
    Editors = {John Earman and John J. Stachel},
    publisher = "University of Minnesota Press",
    year = "1977",
    pages = {61-80}
}
